\documentclass[a4paper,11pt]{article}
\usepackage{bm,mathrsfs,amsmath,amsthm,amssymb,amscd,longtable,array,graphicx}
\newcommand{\nc}{\newcommand}
\nc{\rnc}{\renewcommand}
\nc{\nn}{\nonumber}
\nc{\der}{{\partial}}
\rnc{\Im}{{\rm{Im}\,}}
\rnc{\Re}{{\rm{Re}\,}}
\nc{\db}{\displaybreak[0]\\}
\nc{\bra}{\langle}
\nc{\ket}{\rangle}
\nc{\bs}{\boldsymbol}

\newtheorem{theorem}{Theorem}[section]
\newtheorem{lemma}[theorem]{Lemma}

\newtheorem{proposition}[theorem]{Proposition}

\theoremstyle{definition}
\newtheorem{definition}[theorem]{Definition}

\numberwithin{equation}{section}

\numberwithin{equation}{section}

\begin{document}%
%
\title{Izergin-Korepin analysis on
the wavefunctions of \\
the $U_q(sl_2)$ six-vertex model with reflecting end}

\author{
Kohei Motegi \thanks{E-mail: kmoteg0@kaiyodai.ac.jp}
\\\\
{\it Faculty of Marine Technology, Tokyo University of Marine Science and Technology,}\\
 {\it Etchujima 2-1-6, Koto-Ku, Tokyo, 135-8533, Japan} \\
\\\\
\\
}

\date{\today}

\maketitle

\begin{abstract}
We extend the recently developed Izergin-Korepin analysis on the wavefunctions
of the $U_q(sl_2)$ six-vertex model to the reflecting boundary conditions.
Based on the Izergin-Korepin analysis, we 
determine the exact forms of the symmetric functions which represent
the wavefunctions and its dual.
Comparison of the symmetric functions
with the coordinate Bethe ansatz wavefunctions
for the open XXZ chain by Alcaraz-Barber-Batchelor-Baxter-Quispel
is also made.
As an application, we derive algebraic identities for the symmetric functions
by combining the results
with the determinant formula of the domain wall boundary partition function
of the six-vertex model with reflecting end.

\end{abstract}

\section{Introduction}
Partition functions are fundamental objects in statistial physics
and field theory. In the field of integrable models
\cite{Baxter,KBI,Reshetikhin}, exact computations
of partition functions is one of the most challenging and interesting tasks.
The domain wall boundary partition functions is one of the most well-studied class of partition functions. It was first introduced and investigated
by Korepin \cite{Ko}, and later Izergin found its determinant representation
\cite{Iz} based on his analysis, which have been used for applications to the enumeration of the alternating sign matrices \cite{Br,Ku1,Ku2,Okada} in later years.
The most important step for the analysis of the domain wall boundary
partition functions was the work by Korepin \cite{Ko}, in which he presented a way
how to view partition functions as multivariable polynomials
of spectral parameters by using the quantum inverse scattering method,
which was crucial for the Izergin-Korepin determinant formula \cite{Iz} to be found.
The Izergin-Korepin analysis was applied to various models and variants of
the domain wall boundary partition functions
\cite{Tsuchiya,ZZfelderhof,FCWZfelderhof,FWZfelderhof,PRS,Ros,FK,Ga}
such as the
mixture with the reflecting boundary, half-turn boundary,
and recently extended to the scalar products by Wheeler \cite{Wheeler} which he succeeded
by introducing the notion of intermediate scalar products.

Recently, we extended the Izergin-Korepin analysis to the (projected)
wavefunctions
\cite{Moellipticfelderhof,MoAMP,Motr},
which is a class of partition functions
including the domain wall partition functions as a special case.
Recently, there are extensive studies on the wavefunctions 
(see
\cite{Bogo,BW,BWZ,WZnew,vDE,MS,Motegi,MS2,Korff,GK2,Borodin,BP1,TakeyamaHecke,Takeyama,WZ,BBF,Iv,BBCG,Tabony,BMN,BSfelderhof,LMP}
for examples on which various methods are developed for various
models, boundary conditions, etc),
since it has been widely recognized that the wavefunctions are
integrable model representations of symmetric functions
which is one of the most important objects in representation theory
and algebraic combinatorics, and in many cases it has connections
with other branches of mathematics such as the Schubert calculus
and automorphic representation theory
\cite{LS,FoKi,Buch,IN2,BBBfelderhof}.
We recently constructed the
Izergin-Korepin method to analyze the wavefunctions
of the six-vertex type models \cite{Moellipticfelderhof,MoAMP,Motr},
which is a natural extension
of the one on the domain wall boundary partition functions \cite{Ko,Iz}.
It seems to be a rather universal method in the sense that
it can be used to study exotic boundary conditions,
and can also be extended to the elliptic models
by using the notion of elliptic polynomials.
For example, we analyzed the Deguchi-Martin model
\cite{DMfelderhof}
in \cite{Moellipticfelderhof} and showed that the partition functions
are expressed as a product of
elliptic Schur functions and deformed elliptic Vandermonde determinant.

In this paper, we extend the Izergin-Korepin analysis on the
wavefunctions of the $U_q(sl_2)$ six-vertex model to the reflecting boundary conditions.
As for the domain wall boundary partition functions with reflecting end,
the determinant formula was found by Tsuchiya \cite{Tsuchiya}
(see also Kuperberg \cite{Ku1,Ku2} and Okada \cite{Okada}),
and the thermodynamic limit is investigated by Ribeiro-Korepin \cite{RK}
following the idea of Korepin-Zinn-Justin \cite{KZ}.
As for the wavefunctions under reflecting boundary,
there are studies on the reduced five-vertex model \cite{CMRV}
and related $q$-boson model \cite{WZnew,vDE}, boundary perimeter Bethe ansatz of the XXX chain \cite{Fra},
and solutions of the boundary $q$KZ equation \cite{RSV}.
We show in this paper that the Izergin-Korepin analysis
can be applied to the wavefunctions of the
six vertex model with reflecting end.
Based on the analysis, we determine the exact form of the symmetric functions representing
the wavefunctions and its dual.
As an application of the correspondence between the wavefunctions
and the symmetric functions, we derive algebraic identities
for the symmetric functions by combining with the determinant formula
of Tsuchiya \cite{Tsuchiya} and Kuperberg \cite{Ku2}.
We also compare the homogeneous limit of the symmetric functions
with the coordinate Bethe ansatz wavefunctions
for the open XXZ chain by Alcaraz-Barber-Batchelor-Baxter-Quispel
\cite{ABBBQ}.

This paper is organized as follows.
In the next section, we introduce the wavefunctions and its dual
of the $U_q(sl_2)$ six-vertex model with reflecting end.
In section 3, we present the computation of the simplest case.
In section 4, we perform the Izergin-Korepin analysis
which uniquely characterizes the wavefunctions.
In section 5, we present the explicit forms of the symmetric functions
which represents the wavefunctions, and show that
it satisfies all the properties of the proposition in the section 4.
We compare the symmetric functions with
the coordinate Bethe ansatz wavefunctions
for the open XXZ chain at the end of the section.
As an application,
we derive in section 6 algebraic identities for the symmetric functions
by comparing two ways of evaluations of the domain wall boundary partition functions. Section 7 is devoted to conclusion.

\section{The six-vertex model and the wavefunctions under reflecting boundary}

In this section, we formulate the wavefunctions of the $U_q(sl_2)$
six-vertex model under reflecting boundary and its dual,
which we analyze in this paper.
They can be regarded as natural extensions of the domain wall boundary
partition functions under reflecting boundary in the paper of Kuperberg \cite{Ku2}
(see also Tsuchiya \cite{Tsuchiya}). Note that the $R$-matrices used in this paper
is a slightly gauge transformed one in Kuperberg. We do this
gauge transformation since it is better suited for the Izergin-Korepin analysis on the wavefunctions.

We first introduce two-dimensional Fock spaces $V_a$ and $\mathcal{F}_j$,
$j=1,\dots,M$.
We denote the orthonormal basis of
$V_a$ and its dual as $\{|0 \rangle_a, |1 \rangle_a \}$ and $\{{}_a \langle 0|, {}_a \langle 1|\}$. Similarly, we denote the basis of $\mathcal{F}_j$
and it dual as
$\{|0 \rangle_j, |1 \rangle_j \}$ and $\{{}_j \langle 0|, {}_j \langle 1|\}$.
We usually call $V_a$ as the auxiliary space and $\mathcal{F}_j$
as quantum spaces.

Next, we introduce the $L$-operator of the six-vertex model.
The $L$-operator we use in this paper
is the $U_q(sl_2)$ $R$-matrix \cite{Dr,J}.
We denote the $L$-operator acting on
the spaces $V_a \otimes \mathcal{F}_j$ by $L_{aj}(z,w_j)$,
whose non-zero matrix elements are given by (Figure \ref{pictureloperator})
\begin{align}
{}_a \langle 0| {}_j \langle 0 | L_{a j}(z,w_j)
|0 \rangle_a | 0 \rangle_j&=az^{-1}w_j-a^{-1}z, \\
{}_a \langle 0| {}_j \langle 1 | L_{a j}(z,w_j)
|0 \rangle_a | 1 \rangle_j&=az-a^{-1}z^{-1}w_j, \\
{}_a \langle 0| {}_j \langle 1 | L_{a j}(z,w_j)
|1 \rangle_a | 0 \rangle_j&=a^2-a^{-2}, \\
{}_a \langle 1| {}_j \langle 0 | L_{a j}(z,w_j)
|0 \rangle_a |1 \rangle_j&=(a^2-a^{-2})w_j, \\
{}_a \langle 1| {}_j \langle 0 | L_{a j}(z,w_j)
|1 \rangle_a | 0 \rangle_j&=az-a^{-1}z^{-1}w_j, \\
{}_a \langle 1| {}_j \langle 1 | L_{a j}(z,w_j)
|1 \rangle_a | 1 \rangle_j&=az^{-1}w_j-a^{-1}z.
\end{align}

\begin{figure}[ht]
\includegraphics[width=12cm]{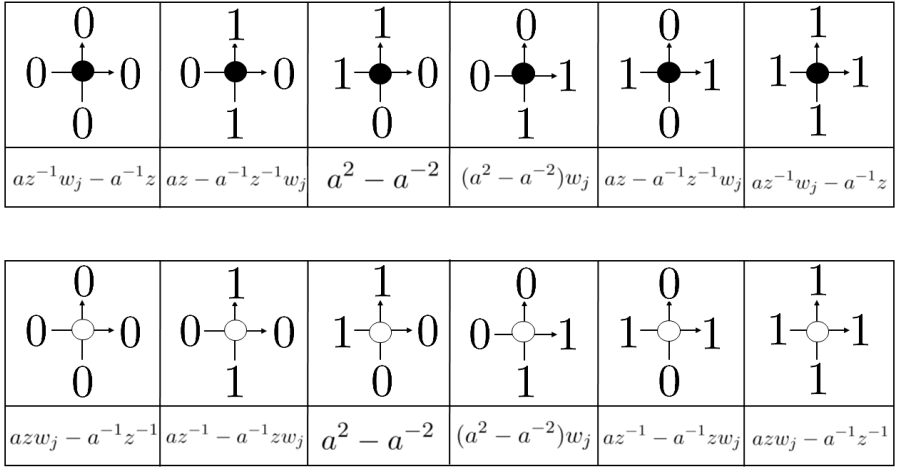}
\caption{A graphical description of the $L$-operators
$L_{aj}(z,w_j)$ (top) and $L_{aj}(z^{-1},w_j)$ (bottom),
used to construct the double-row monodromy matrix.
}
\label{pictureloperator}
\end{figure}

Using the $L$-operators, we construct the monodromy matrix
\begin{align}
T_a(z|w_1,\dots,w_M)=L_{a1}(z,w_1) \cdots L_{aM}(z,w_M).
\end{align}

We also introduce the following $K$-operator $K_{a_2 a_1}(z)$ \cite{Sklyanin}
acting on the auxiliary space (Figure \ref{picturekanddouble} top)
\begin{align}
K_{a_2 a_1}(z)={}_{a_2} \langle 0| {}_{a_1} \langle 1| (baz-b^{-1}a^{-1}z^{-1})
+{}_{a_2} \langle 1| {}_{a_1} \langle 0| (ba^{-1}z^{-1}-b^{-1}az).
\end{align}

\begin{figure}[ht]
\includegraphics[width=12cm]{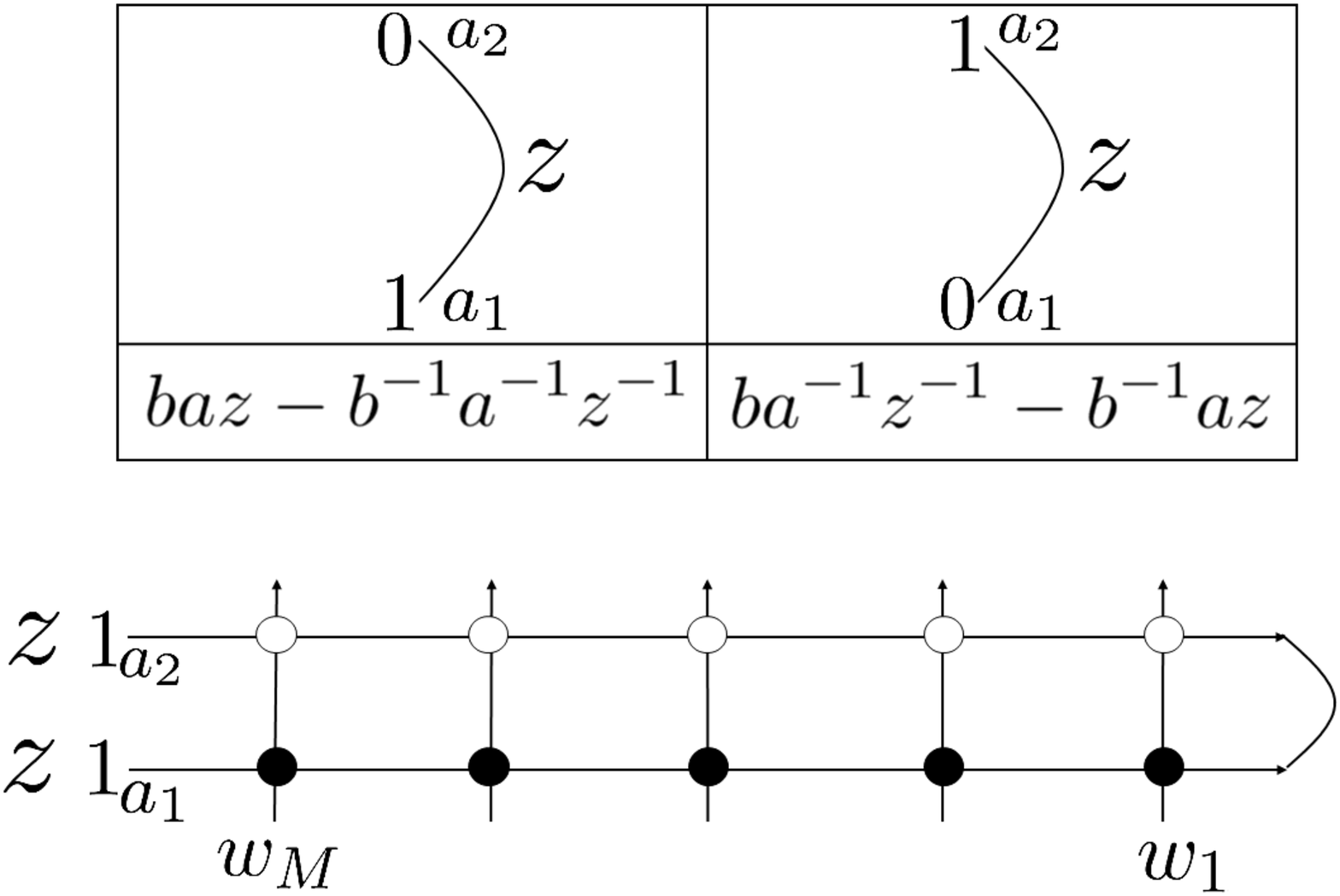}
\caption{A graphical description of the $K$-operator
$K_{a_2 a_1}(z)$ (top) and the double-row $B$-operator $\mathcal{B}(z|w_1,\dots,w_M)$ (bottom).
}
\label{picturekanddouble}
\end{figure}

We now introduce the double-row monodromy matrix
using the monodromy matrix
$T_a(z|w_1,\dots,w_M)$
and the $K$-operator $K_{a_2 a_1}(z)$ as
\begin{align}
\mathcal{T}_{a_2 a_1}(z|w_1,\dots,w_M)
=K_{a_2 a_1}(z) T_{a_2}(z^{-1}|w_1,\dots,w_M)
T_{a_1}(z|w_1,\dots,w_M).
\end{align}
We define the following double-row $B$-operator
(Figure \ref{picturekanddouble} bottom)
as a matrix element
of the double-row monodromy matrix with respect to the auxilary space
\begin{align}
&\mathcal{B}(z|w_1,\dots,w_M)
=\mathcal{T}_{a_2 a_1}(z|w_1,\dots,w_M)|1 \rangle_{a_2} |1 \rangle_{a_1}.
\end{align}

Using the matrix elements of the ordinary monodromy matrices
\begin{align}
B(z|w_1,\dots,w_M)&=
{}_a \langle 0|T_a(z|w_1,\dots,w_M)|1 \rangle_a, \\
D(z|w_1,\dots,w_M)&=
{}_a \langle 1|T_a(z|w_1,\dots,w_M)|1 \rangle_a, 
\end{align}
the double-row $B$-operator is written as
\begin{align}
\mathcal{B}(z|w_1,\dots,w_M)
=&K_{a_2 a_1}(z) T_{a_2}(z^{-1}|w_1,\dots,w_M) T_{a_1}(z|w_1,\dots,w_M)
|1 \rangle_{a_2} |1 \rangle_{a_1} \nonumber \\
=&(ba^{-1}z^{-1}-b^{-1}az)D(z^{-1}|w_1,\dots,w_M)B(z|w_1,\dots,w_M)
\nonumber \\
&+(baz-b^{-1}a^{-1}z^{-1})B(z^{-1}|w_1,\dots,w_M)D(z|w_1,\dots,w_M).
\end{align}

\begin{figure}[ht]
\includegraphics[width=12cm]{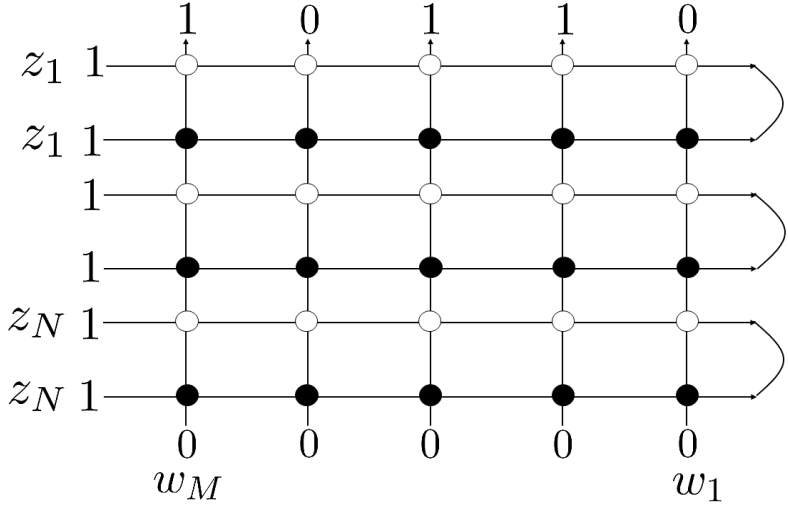}
\caption{A graphical description of the wavefunction under reflecting boundary
$W_{M,N}(z_1,\dots,z_N|w_1,\dots,w_M|x_1,\dots,x_N)$.
The figure illustrates the case $M=5$, $N=3$, $x_1=2$,
$x_2=3$, $x_3=5$.
}
\label{picturewavefunction}
\end{figure}

In order to introduce wavefunctions,  we also define special states in the tensor product of the Fock spaces
$\mathcal{F}_1 \otimes \cdots \otimes \mathcal{F}_M$
and its dual $(\mathcal{F}_1 \otimes \cdots \otimes \mathcal{F}_M)^*$
as
\begin{align}
| 0^{M} \rangle&:=|0\ket_1\
\otimes \dots \otimes |0\ket_M, \\
| 1^{M} \rangle&:=|1\ket_1\
\otimes \dots \otimes |1\ket_M, \\
\langle 0^{M}|&:=
{}_1\bra 0|\otimes\dots \otimes{}_M\bra 0|, \\
\langle 1^{M}|&:=
{}_1\bra 1|\otimes\dots \otimes{}_M\bra 1|.
\end{align}
By acting the operators $\sigma_j^+$ and $\sigma_j^-$ defined by
\begin{align}
\sigma^+_j|1 \rangle_k=\delta_{jk}|0 \rangle_k, \
\sigma^+_j|0 \rangle_k =0, \
{}_k \langle 1|\sigma^+_j=0, \
{}_k \langle 0|\sigma^+_j=\delta_{jk} {}_k \langle 1|, \\
\sigma^-_j|0 \rangle_k=\delta_{jk}|1 \rangle_k, \
\sigma^-_j|1 \rangle_k =0, \
{}_k \langle 0|\sigma^-_j=0, \
{}_k \langle 1|\sigma^-_j=\delta_{jk} {}_k \langle 0|,
\end{align}
on $|0^M \rangle$, $\langle 0^M|$ and $\langle 1^M|$,
we introduce states
\begin{align}
\langle x_1 \cdots x_N|
&=
\langle 0^M|
\prod_{j=1}^N \sigma^+_{x_j},
\label{particleconfiguration} \\
|x_1 \cdots x_N \rangle
&=
\prod_{j=1}^N \sigma^-_{x_j}|0^M \rangle,
\end{align}
for integers $x_1,\dots,x_N$ satisfying
$1 \le x_1 < x_2 < \cdots < x_N \le M$,
and
\begin{align}
|\overline{x_1} \cdots \overline{x_N} \rangle
&=
\prod_{j=1}^N \sigma^+_{\overline{x_j}}|1^M \rangle,
\label{dualholeconfiguration}
\end{align}
for integers $\overline{x_1},\dots,\overline{x_N}$
satisfying
$1 \le \overline{x_1}< \overline{x_2} < \cdots < \overline{x_N} \le M$.

\begin{figure}[ht]
\includegraphics[width=12cm]{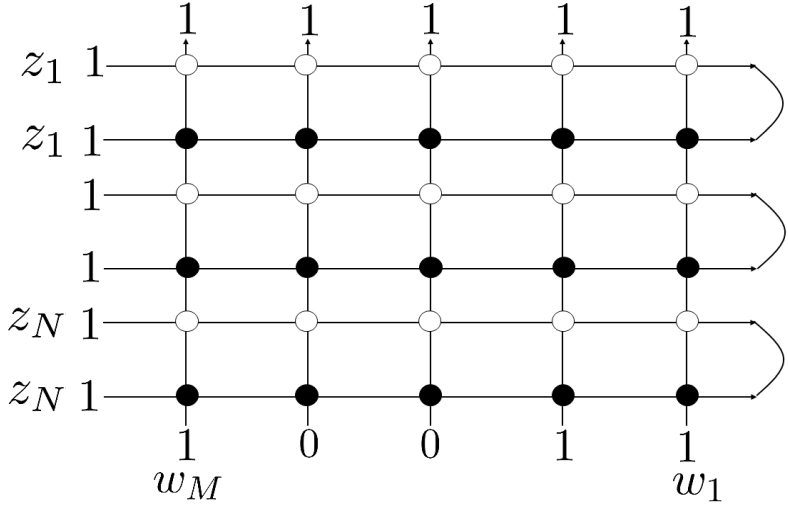}
\caption{A graphical description of the dual wavefunction under reflecting boundary
$\overline{W}_{M,N}(z_1,\dots,z_N|w_1,\dots,w_M|\overline{x_1},\dots,\overline{x_N})$.
The figure illustrates the case $M=5$, $N=3$, $\overline{x_1}=1$,
$\overline{x_2}=2$, $\overline{x_3}=5$.
}
\label{picturedualwavefunction}
\end{figure}

We define (projected) wavefunctions $W_{M,N}(z_1,\dots,z_N|w_1,\dots,w_M|x_1,\dots,x_N)$
by acting the double-row $B$-operators $\mathcal{B}(z_j|w_1,\dots,w_M)$ ($j=1,\dots,N$) on the state $|0^M \rangle$
and projecting to the state
$\langle x_1 \cdots x_N|$ (Figure \ref{picturewavefunction})
\begin{align}
&W_{M,N}(z_1,\dots,z_N|w_1,\dots,w_M|x_1,\dots,x_N) \nonumber \\
&=\langle x_1 \cdots x_N|\mathcal{B}(z_1|w_1,\dots,w_M) \cdots
\mathcal{B}(z_N|w_1,\dots,w_M)|0^M \rangle.
\end{align}
Similarly, we define the dual wavefunction
$\overline{W}_{M,N}(z_1,\dots,z_N|w_1,\dots,w_M|x_1,\dots,x_N)$ as
(Figure \ref{picturedualwavefunction})
\begin{align}
&\overline{W}_{M,N}(z_1,\dots,z_N|w_1,\dots,w_M|\overline{x_1},\dots,\overline{x_N}) \nonumber \\
&=\langle 1^M|\mathcal{B}(z_1|w_1,\dots,w_M) \cdots
\mathcal{B}(z_N|w_1,\dots,w_M)|\overline{x_1} \cdots \overline{x_N} \rangle.
\end{align}

\section{The simplest case}
In this section, we examine the simplest case $N=1$.
A special case of the result obtained in this section will be used
in the next section as the initial condition
for the Izergin-Korepin analysis on the wavefunctions under
reflecting boundary. First, we prepare the following lemma.

\begin{lemma} \label{identitylemma}
The following identity holds
\begin{align}
&(a^2-a^{-2})(z^2-z^{-2})\sum_{j=1}^{x_1-1}w_j
\prod_{k=1}^{j-1}(az^{-1}-a^{-1}zw_k)(az^{-1}w_k-a^{-1}z) \nonumber \\
\times&\prod_{k=j+1}^{x_1-1}(azw_k-a^{-1}z^{-1})(az-a^{-1}z^{-1}w_k) \nonumber
\\
=&\prod_{k=1}^{x_1-1}(azw_k-a^{-1}z^{-1})(az-a^{-1}z^{-1}w_k)
-\prod_{k=1}^{x_1-1}(az^{-1}w_k-a^{-1}z)(az^{-1}-a^{-1}z w_k)
. \label{identity}
\end{align}
\end{lemma}

\begin{proof}
This can be proved by induction on $x_1$.
It is easy to check the case $x_1=2$.
Let us assume \eqref{identity} holds
and show that the identity \eqref{identity} with $x_1$
replaced by $x_1+1$ also holds.
We first decompose the left hand side of \eqref{identity}
with $x_1$ replaced by $x_1+1$ as
\begin{align}
&(a^2-a^{-2})(z^2-z^{-2})\sum_{j=1}^{x_1}w_j
\prod_{k=1}^{j-1}(az^{-1}-a^{-1}zw_k)(az^{-1}w_k-a^{-1}z)
\nonumber \\
\times&\prod_{k=j+1}^{x_1}(azw_k-a^{-1}z^{-1})(az-a^{-1}z^{-1}w_k) \nonumber
\\
=&(a^2-a^{-2})(z^2-z^{-2})w_{x_1}\prod_{k=1}^{x_1-1}(az^{-1}-a^{-1}zw_k)(az^{-1}w_k-a^{-1}z) \nonumber \\
&+(azw_{x_1}-a^{-1}z^{-1})(az-a^{-1}z^{-1}w_{x_1})(a^2-a^{-2})(z^2-z^{-2})
\sum_{j=1}^{x_1-1}w_j \nonumber \\
&\times \prod_{k=1}^{j-1}(az^{-1}-a^{-1}zw_k)(az^{-1}w_k-a^{-1}z)
\prod_{k=j+1}^{x_1-1}(azw_k-a^{-1}z^{-1})(az-a^{-1}z^{-1}w_k).
\label{forrearrange}
\end{align}
We next apply the assumption \eqref{identity} to
the right hand side of \eqref{forrearrange} and rearrange as follows:
\begin{align}
&(a^2-a^{-2})(z^2-z^{-2})\sum_{j=1}^{x_1}w_j
\prod_{k=1}^{j-1}(az^{-1}-a^{-1}zw_k)(az^{-1}w_k-a^{-1}z)
\nonumber \\
\times&\prod_{k=j+1}^{x_1}(azw_k-a^{-1}z^{-1})(az-a^{-1}z^{-1}w_k) \nonumber
\\
=&(a^2-a^{-2})(z^2-z^{-2})w_{x_1}\prod_{k=1}^{x_1-1}(az^{-1}-a^{-1}zw_k)(az^{-1}w_k-a^{-1}z) \nonumber \\
&+(azw_{x_1}-a^{-1}z^{-1})(az-a^{-1}z^{-1}w_{x_1})
\nonumber \\
&\times \Bigg(
\prod_{k=1}^{x_1-1}(azw_k-a^{-1}z^{-1})(az-a^{-1}z^{-1}w_k)
-\prod_{k=1}^{x_1-1}(az^{-1}w_k-a^{-1}z)(az^{-1}-a^{-1}z w_k)
\Bigg) \nonumber 
\end{align}
\begin{align}
=&\prod_{k=1}^{x_1}(azw_k-a^{-1}z^{-1})(az-a^{-1}z^{-1}w_k)
-\prod_{k=1}^{x_1-1}(az^{-1}w_k-a^{-1}z)(az^{-1}-a^{-1}z w_k)
\nonumber \\
&\times \{ (azw_{x_1}-a^{-1}z^{-1})(az-a^{-1}z^{-1}w_{x_1})
-(a^2-a^{-2})(z^2-z^{-2})w_{x_1} \} \nonumber \\
=&\prod_{k=1}^{x_1}(azw_k-a^{-1}z^{-1})(az-a^{-1}z^{-1}w_k)
-\prod_{k=1}^{x_1}(az^{-1}w_k-a^{-1}z)(az^{-1}-a^{-1}z w_k).
\end{align}
Hence we find the identity \eqref{identity} with $x_1$
replaced by $x_1+1$ also holds.

\end{proof}

\begin{proposition} \label{simplestproposition}
The wavefunction $W_{M,1}(z|w_1,\dots,w_M|x_1)$
is explicitly expressed as
\begin{align}
&W_{M,1}(z|w_1,\dots,w_M|x_1)
=(a^2-a^{-2})(-a^2 z^2+a^{-2} z^{-2})
\sum_{\tau=\pm 1} \frac{1}{z^{2 \tau}-z^{-2 \tau}}
(ba^{-1}z^{-\tau}-b^{-1}az^\tau) \nonumber \\
&\times \prod_{j=1}^M (az^{-\tau}-a^{-1}z^\tau w_j)
\prod_{j=1}^{x_1-1}(az^{-\tau}w_j-a^{-1}z^\tau)
\prod_{j=x_1+1}^M (az^\tau-a^{-1}z^{-\tau}w_j). \label{simplestmatrixelement}
\end{align}
\end{proposition}

\begin{proof}
Decomposing $W_{M,1}(z|w_1,\dots,w_M|x_1)$ as
\begin{align}
&W_{M,1}(z|w_1,\dots,w_M|x_1) \nonumber \\
=&(baz-b^{-1}a^{-1}z^{-1})
\langle x_1 |B(z^{-1}|w_1,\dots,w_M)| 0^M \rangle
\langle 0^M |D(z|w_1,\dots,w_M)| 0^M \rangle \nonumber \\
&+(ba^{-1}z^{-1}-b^{-1}az)
\langle x_1 |D(z^{-1}|w_1,\dots,w_M)| x_1 \rangle
\langle x_1 |B(z|w_1,\dots,w_M)| 0^M \rangle \nonumber \\
&+(ba^{-1}z^{-1}-b^{-1}az)
\sum_{j=1}^{x_1-1}
\langle x_1 |D(z^{-1}|w_1,\dots,w_M)| j \rangle
\langle j |B(z|w_1,\dots,w_M)| 0^M \rangle,
\end{align}
and using the matrix elements
\begin{align}
&\langle 0^M |D(z|w_1,\dots,w_M)| 0^M \rangle
=\prod_{k=1}^M (az-a^{-1}z^{-1}w_k), \nonumber \\
&\langle x_1 |B(z|w_1,\dots,w_M)| 0^M \rangle
=(a^2-a^{-2}) \prod_{k=x_1+1}^M (az-a^{-1}z^{-1}w_k) \prod_{k=1}^{x_1-1}(az^{-1}w_k-a^{-1}z), \nonumber \\
&\langle j |B(z|w_1,\dots,w_M)| 0^M \rangle \nonumber \\
=&(a^2-a^{-2})\prod_{k=j+1}^M (az-a^{-1}z^{-1}w_k) \prod_{k=1}^{j-1}(az^{-1}w_k-a^{-1}z), \ j=1,\dots,x_1-1, \nonumber \\
&\langle x_1 |B(z^{-1}|w_1,\dots,w_M)| 0^M \rangle
=(a^2-a^{-2})\prod_{k=x_1+1}^M (az^{-1}-a^{-1}zw_k) \prod_{k=1}^{x_1-1}(azw_k-a^{-1}z^{-1}), \nonumber \\
&\langle x_1 |D(z^{-1}|w_1,\dots,w_M)| x_1 \rangle
=(azw_{x_1}-a^{-1}z^{-1}) \prod_{k=x_1+1}^M (az^{-1}-a^{-1}zw_k) \prod_{k=1}^{x_1-1}(az^{-1}-a^{-1}zw_k), \nonumber 
\end{align}
\begin{align}
&\langle x_1 |D(z^{-1}|w_1,\dots,w_M)| j \rangle
=(a^2-a^{-2})^2 w_j \prod_{k=j+1}^{x_1-1}(azw_k-a^{-1}z^{-1}) \nonumber \\
&\times \prod_{k=x_1+1}^M (az^{-1}-a^{-1}zw_k) \prod_{k=1}^{j-1}(az^{-1}-a^{-1}zw_k)
, \ j=1,\dots,x_1-1, \nonumber
\end{align}
we can explicitly calculate
$W_{M,1}(z|w_1,\dots,w_M|x_1)$ as
\begin{align}
&W_{M,1}(z|w_1,\dots,w_M|x_1)
=(a^2-a^{-2})\prod_{k=x_1+1}^{M} (az^{-1}-a^{-1}zw_k)(az-a^{-1}z^{-1}w_k)
\nonumber \\
&\times \Bigg\{ (baz-b^{-1}a^{-1}z^{-1})
\prod_{k=1}^{x_1-1}(azw_k-a^{-1}z^{-1})
\prod_{k=1}^{x_1}(az-a^{-1}z^{-1}w_k) \nonumber \\
&+(ba^{-1}z^{-1}-b^{-1}az)(azw_{x_1}-a^{-1}z^{-1})
\prod_{k=1}^{x_1-1}(az^{-1}-a^{-1}zw_k)
(az^{-1}w_k-a^{-1}z) \nonumber \\
&+(a^2-a^{-2})^2 \sum_{j=1}^{x_1-1}w_j (ba^{-1}z^{-1}-b^{-1}az)
\prod_{k=1}^{j-1} (az^{-1}-a^{-1}zw_k)(az^{-1}w_k-a^{-1}z) \nonumber \\
&\times
\prod_{k=j+1}^{x_1-1}(azw_k-a^{-1}z^{-1})
\prod_{k=j+1}^{x_1}(az-a^{-1}z^{-1}w_k)
\Bigg\}. \label{rhssimplest}
\end{align}
One can show by tedious but straightforward computation that
proving the right hand side of \eqref{rhssimplest}
is equal to
\begin{align}
&(a^2-a^{-2})(-a^2 z^2+a^{-2} z^{-2})
\sum_{\tau=\pm 1} \frac{1}{z^{2 \tau}-z^{-2 \tau}}
(ba^{-1}z^{-\tau}-b^{-1}az^\tau) \nonumber \\
&\times \prod_{j=1}^M (az^{-\tau}-a^{-1}z^\tau w_j)
\prod_{j=1}^{x_1-1}(az^{-\tau}w_j-a^{-1}z^\tau)
\prod_{j=x_1+1}^M (az^\tau-a^{-1}z^{-\tau}w_j),
\end{align}
reduces to proving the equality \eqref{identity}
in Lemma \ref{identitylemma}, which we have already proved.
\end{proof}
The dual wavefunction $\overline{W}_{M,1}(z|w_1,\dots,w_M|\overline{x_1})$
can be computed in the same way.
We state the result below.
\begin{proposition}
The dual wavefunction $\overline{W}_{M,1}(z|w_1,\dots,w_M|\overline{x_1})$
is explicitly expressed as
\begin{align}
&\overline{W}_{M,1}(z|w_1,\dots,w_M|\overline{x_1})
=(a^2-a^{-2})(-a^2 z^2+a^{-2} z^{-2})
\sum_{\tau=\pm 1} \frac{1}{z^{2 \tau}-z^{-2 \tau}}
(baz^{\tau}-b^{-1}a^{-1}z^{-\tau}) \nonumber \\
&\times \prod_{j=1}^M (az^{-\tau}w_j-a^{-1}z^\tau)
\prod_{j=1}^{\overline{x_1}-1}(az^{-\tau}-a^{-1}z^\tau w_j)
\prod_{j=\overline{x_1}+1}^M (az^\tau w_j-a^{-1}z^{-\tau}).
\end{align}
\end{proposition}

\section{Izergin-Korepin analysis}
In this section, we perform the Izergin-Korepin analysis \cite{Ko,Iz}
which uniquely characterizes the wavefunctions.
See \cite{Motr} for a simpler case of the Izergin-Korepin anlaysis
on the wavefunctions without reflecting boundary.
The Proposition given below is the extension to the
reflecting boundary condition.

\begin{proposition}
\label{ordinarypropertiesfordomainwallboundarypartitionfunction}
The wavefunctions
$W_{M,N}(z_1,\dots,z_N|w_1,\dots,w_M|x_1,\dots,x_N)$
satisfies the following properties. \\
\\
 (1) When $x_N=M$, the wavefunctions
$W_{M,N}(z_1,\dots,z_N|w_1,\dots,w_M|x_1,\dots,x_N)$
is a polynomial of degree $2N-1$ in $w_M$.
\\
 (2) The wavefunctions $W_{M,N}(z_1,\dots,z_N|w_1,\dots,w_M|x_1,\dots,x_N)$
is symmetric with respect to $z_1,\dots,z_N$, i.e.,
\begin{align}
W_{M,N}(z_1,\dots,z_N|w_1,\dots,w_M|x_1,\dots,x_N)
=W_{M,N}(z_{\sigma(1)},\dots,z_{\sigma(N)}|w_1,\dots,w_M|x_1,\dots,x_N),
\label{symmetrywavefunction}
\end{align}
for $\sigma \in S_N$.
\\
(3)
The wavefunctions $W_{M,N}(z_1,\dots,z_N|w_1,\dots,w_M|x_1,\dots,x_N)$
with $z_i$ replaced by $z_i^{-1}$ is connected with the original one
by
\begin{align}
\frac{W_{M,N}(z_1,\dots,z_N|w_1,\dots,w_M|x_1,\dots,x_N)|_{
z_i \longleftrightarrow z_i^{-1}}}
{W_{M,N}(z_1,\dots,z_N|w_1,\dots,w_M|x_1,\dots,x_N)}
=\frac{a^2 z_i^{-2}-a^{-2}z_i^2}{a^2 z_i^2-a^{-2} z_i^{-2}}.
\label{permutationwavefunction}
\end{align}
\\
(4) The following recursive relations between the
wavefunctions hold if $x_N=M$
(Figure \ref{izerginkorepinone}):
\begin{align}
&W_{M,N}(z_1,\dots,z_N|w_1,\dots,w_M|x_1,\dots,x_N)
|_{w_M=a^2 z_N^2}
\nonumber \\
=&(a^2-a^{-2})(ba^{-1}z_N^{-1}-b^{-1}az_N)
\prod_{j=1}^{N}(a^3 z_N^2 z_j-a^{-1}z_j^{-1})
\prod_{j=1}^{N-1} (a^3 z_N^2 z_j^{-1}-a^{-1}z_j) \nonumber \\
&\times \prod_{j=1}^{M-1} (az_N^{-1}w_j-a^{-1}z_N)(az_N^{-1}-a^{-1}z_N w_j)
\nonumber \\
&\times W_{M-1,N-1}(z_1,\dots,z_{N-1}|w_1,\dots,w_{M-1}|x_1,\dots,x_{N-1})
. \label{ordinaryrecursionwavefunction}
\end{align}

If $x_N \neq M$, the following factorizations hold for the wavefunctions
(Figure \ref{izerginkorepintwo}):
\begin{align}
&W_{M,N}(z_1,\dots,z_N|w_1,\dots,w_M|x_1,\dots,x_N)
 \nonumber \\
=&\prod_{j=1}^N (az_j^{-1}-a^{-1}z_j w_M)(az_j-a^{-1}z_j^{-1}w_M)
W_{M-1,N}(z_1,\dots,z_N|w_1,\dots,w_{M-1}|x_1,\dots,x_N).
\label{ordinaryrecursionwavefunction2}
\end{align}
\\
(5) The following holds for the case $N=1$, $x_1=M$
\begin{align}
&W_{M,1}(z|w_1,\dots,w_M|M)
=(a^2-a^{-2})(-a^2 z^2+a^{-2} z^{-2})
\sum_{\tau=\pm 1} \frac{1}{z^{2 \tau}-z^{-2 \tau}}
(ba^{-1}z^{-\tau}-b^{-1}az^\tau) \nonumber \\
&\times \prod_{j=1}^M (az^{-\tau}-a^{-1}z^\tau w_j)
\prod_{j=1}^{M-1}(az^{-\tau}w_j-a^{-1}z^\tau).
\label{ordinaryinitialrecursion}
\end{align}
\end{proposition}

\begin{figure}[ht]
\includegraphics[width=12cm]{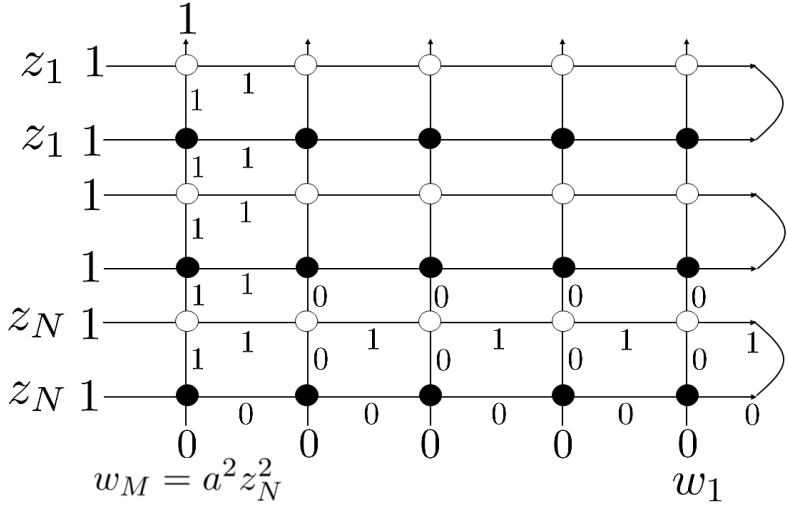}
\caption{A graphical description of the relation
\eqref{ordinaryrecursionwavefunction}.
By using the ice-rule,
one can see that if one sets $w_M$ to $w_M=a^2 z_N^2$,
all the $L$-operators at the
leftmost column and the bottom double-row
freeze.
}
\label{izerginkorepinone}
\end{figure}

\begin{proof}
This can be proved in the standard way.
Properties (2) and (3) can be proved in the same way with
Tsuchiya \cite{Tsuchiya} and Kuperberg \cite{Ku2}.
Note that the state $\langle x_1,\dots,x_N|$ which is used for the
constuction of the wavefunctions
$W_{M,N}(z_1,\dots,z_N|w_1,\dots,w_M|x_1,\dots,x_N)$
do not have any affect on the arguments.
Property (3) can also be proved from Property (2) in the following way,
which have originally appeared in the discussion of another type of
six-vertex model \cite{dualsymplecticfelderhof}.
We insert the completeness relation to get
\begin{align}
&W_{M,N}(z_1,\dots,z_N|w_1,\dots,w_M|x_1,\dots,x_N) \nonumber \\
=&\sum_{x} \langle x_1,\dots,x_N|
\mathcal{B}(z_1|w_1,\dots,w_M) \cdots \mathcal{B}(z_{N-1}|w_1,\dots,w_M)| x \rangle
\langle x |\mathcal{B}(z_N|w_1,\dots,w_M)|0^M \rangle \nonumber \\
=&\sum_{x} \langle x_1,\dots,x_N|
\mathcal{B}(z_1|w_1,\dots,w_M) \cdots \mathcal{B}(z_{N-1}|w_1,\dots,w_M)| x \rangle
W_{M,1}(z_N|w_1,\dots,w_M|x). \label{decompositionforinversion}
\end{align}
From the explicit expression of $W_{M,1}(z_N|w_1,\dots,w_M|x)$
\eqref{simplestmatrixelement}, one can see
\begin{align}
\frac{W_{M,1}(z_N|w_1,\dots,w_M|x)|_{z_N \longleftrightarrow z_N^{-1}}}
{W_{M,1}(z_N|w_1,\dots,w_M|x)}
=\frac{a^2 z_N^{-2}-a^{-2} z_N^2}{a^2 z_N^2-a^{-2} z_N^{-2}}.
\label{simpleratio}
\end{align}
Since the ratio \eqref{simpleratio} does not depend on $x$,
combining \eqref{decompositionforinversion} and \eqref{simpleratio} gives
\begin{align}
\frac{W_{M,N}(z_1,\dots,z_N|w_1,\dots,w_M|x_1,\dots,x_N)|_{
z_N \longleftrightarrow z_N^{-1}}}
{W_{M,N}(z_1,\dots,z_N|w_1,\dots,w_M|x_1,\dots,x_N)}
=\frac{a^2 z_N^{-2}-a^{-2}z_N^2}{a^2 z_N^2-a^{-2} z_N^{-2}}.
\label{ratioforgeneric}
\end{align}
\eqref{permutationwavefunction} follows from \eqref{ratioforgeneric}
since the wavefunction is symmetric with respect to $z_1,\dots,z_N$
(Property (2)).

Property (5) is a special case $x_1=M$ of \eqref{simplestmatrixelement} in
Proposition \ref{simplestproposition} which is proved in the last section.

To show Property (1),
let us look at the leftmost column of the wavefunctions
since the dependence on the parameter $w_M$ comes from the
$L$-operators at this column.
When $x_N=M$,
one can see that we must use at least one of the matrix elements
${}_a \langle 0 | {}_M \langle 1 |  L_{aM}(z,w_M)|1 \rangle_a |0 \rangle_M=a^2-a^{-2}$ or
${}_a \langle 0 | {}_M \langle 1 |  L_{aM}(z^{-1},w_M)|1 \rangle_a |0 \rangle_M=a^2-a^{-2}$ of the $L$-operators
among the $2N$ $L$-operators at the leftmost
column. These matrix elements do not involve $w_M$,
from which it follows that the
wavefunctions $W_{M,N}(z_1,\dots,z_N|w_1,\dots,w_M|x_1,\dots,x_N)$
is a polynomial of degree $2N-1$ in $w_M$.

Property (4) can be easily shown with the help of the graphical
representation of the wavefunctions 
and the ice-rule of the $L$-operators of the six-vertex model
${}_a \langle \gamma  |{}_j \langle \delta | L_{aj}(z,w)|\alpha \rangle_a |\beta \rangle_j=0$ unless $\alpha+\beta=\gamma+\delta$
(Figures \ref{izerginkorepinone} and \ref{izerginkorepintwo}).
One finds that when $x_N=M$, the leftmost column and the bottom row
freeze if one sets $w_M$ to $w_M=a^2 z_N^2$. From the frozen part, one gets the factor
\begin{align}
&(a^2-a^{-2})(ba^{-1}z_N^{-1}-b^{-1}az_N)
\prod_{j=1}^{N}(a^3 z_N^2 z_j-a^{-1}z_j^{-1})
\prod_{j=1}^{N-1} (a^3 z_N^2 z_j^{-1}-a^{-1}z_j) \nonumber \\
&\times \prod_{j=1}^{M-1} (az_N^{-1}w_j-a^{-1}z_N)(az_N^{-1}-a^{-1}z_N w_j),
\end{align}
as the product of the matrix elements of the $L$-operators.
On the other hand, the remaining part is
$W_{M-1,N-1}(z_1,\dots,z_{N-1}|w_1,\dots,w_{M-1}|x_1,\dots,x_{N-1})$,
and one can see that \eqref{ordinaryrecursionwavefunction} follows.

When $x_N \neq M$,
one sees that the leftmost column freezes from which the factor
$\prod_{j=1}^N (az_j^{-1}-a^{-1}z_j w_M)(az_j-a^{-1}z_j^{-1}w_M)$
contributes to the wavefunctions, and the remaining part is nothing but
$W_{M-1,N}(z_1,\dots,z_N|w_1,\dots,w_{M-1}|x_1,\dots,x_N)$.
Hence, for the case $x_N \neq M$, the wavefunctions
$W_{M,N}(z_1,\dots,z_N|w_1,\dots,w_{M}|x_1,\dots,x_N)$
is the product of these two parts and
\eqref{ordinaryrecursionwavefunction2} follows.
\end{proof}

\begin{figure}[ht]
\includegraphics[width=12cm]{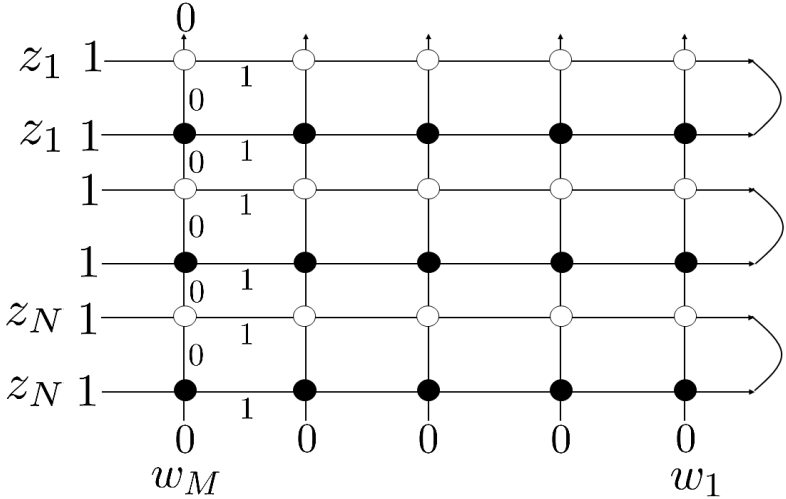}
\caption{A graphical description of the relation
\eqref{ordinaryrecursionwavefunction2}.
From the ice-rule,
one can see that all the $L$-operators at the
leftmost column freeze.
}
\label{izerginkorepintwo}
\end{figure}

At the end of this section, we list the property
for the dual wavefunctions, which can also be proved
in the same way with the wavefunctions.

\begin{proposition}
The dual wavefunctions
$\overline{W}_{M,N}(z_1,\dots,z_N|w_1,\dots,w_M|\overline{x_1},\dots,\overline{x_N})$
satisfies the following properties. \\
\\
 (1) When $\overline{x_N}=M$, the dual wavefunctions
$\overline{W}_{M,N}(z_1,\dots,z_N|w_1,\dots,w_M|\overline{x_1},\dots,\overline{x_N})$
is a polynomial of degree $2N-1$ in $w_M$.
\\
 (2) The dual wavefunctions $\overline{W}_{M,N}(z_1,\dots,z_N|w_1,\dots,w_M|\overline{x_1},\dots,\overline{x_N})$
is symmetric with respect to $z_1,\dots,z_N$, i.e.,
\begin{align}
\overline{W}_{M,N}(z_1,\dots,z_N|w_1,\dots,w_M|\overline{x_1},\dots,\overline{x_N})
=\overline{W}_{M,N}(z_{\sigma(1)},\dots,z_{\sigma(N)}|w_1,\dots,w_M|\overline{x_1},\dots,\overline{x_N}),
\end{align}
for $\sigma \in S_N$.
\\
(3) The dual wavefunctions $\overline{W}_{M,N}(z_1,\dots,z_N|w_1,\dots,w_M|\overline{x_1},\dots,\overline{x_N})$
with $z_i$ replaced by $z_i^{-1}$ is connected with the original one
by
\begin{align}
\frac{\overline{W}_{M,N}(z_1,\dots,z_N|w_1,\dots,w_M|\overline{x_1},\dots,\overline{x_N})|_{
z_i \longleftrightarrow z_i^{-1}}}
{\overline{W}_{M,N}(z_1,\dots,z_N|w_1,\dots,w_M|\overline{x_1},\dots,\overline{x_N})}
=\frac{a^2 z_i^{-2}-a^{-2}z_i^2}{a^2 z_i^2-a^{-2} z_i^{-2}}.
\end{align}
\\
(4) The following recursive relations between the
dual wavefunctions hold if $\overline{x_N}=M$:
\begin{align}
&\overline{W}_{M,N}(z_1,\dots,z_N|w_1,\dots,w_M|\overline{x_1},\dots,\overline{x_N})|_{w_M=a^{-2} z_N^{-2}}
\nonumber \\
=&(a^2-a^{-2})(baz_N-b^{-1}a^{-1}z_N^{-1})
\prod_{j=1}^{N}(a z_j-a^{-3}z_j^{-1} z_N^{-2})
\prod_{j=1}^{N-1} (az_j^{-1}-a^{-3}z_jz_N^{-2}) \nonumber \\
&\times \prod_{j=1}^{M-1} (az_N^{-1}-a^{-1}z_Nw_j)
(az_N^{-1}w_j-a^{-1}z_N)
\nonumber \\
&\times \overline{W}_{M-1,N-1}(z_1,\dots,z_{N-1}|w_1,\dots,w_{M-1}|
\overline{x_1},\dots,\overline{x_{N-1}})
.
\end{align}

If $\overline{x_N} \neq M$, the following factorizations hold for the dual wavefunctions:
\begin{align}
&\overline{W}_{M,N}(z_1,\dots,z_N|w_1,\dots,w_M|\overline{x_1},\dots,\overline{x_N})
 \nonumber \\
=&\prod_{j=1}^N (az_jw_M-a^{-1}z_j^{-1})(az_j^{-1}w_M-a^{-1}z_j)
\overline{W}_{M-1,N}(z_1,\dots,z_N|w_1,\dots,w_{M-1}|\overline{x_1},\dots,\overline{x_N}).
\end{align}
\\
(5) The following holds for the case $N=1$, $\overline{x_1}=M$
\begin{align}
&\overline{W}_{M,1}(z|w_1,\dots,w_M|M)
=(a^2-a^{-2})(-a^2 z^2+a^{-2} z^{-2})
\sum_{\tau=\pm 1} \frac{1}{z^{2 \tau}-z^{-2 \tau}}
(baz^{\tau}-b^{-1}a^{-1}z^{-\tau}) \nonumber \\
&\times \prod_{j=1}^M (az^{-\tau}w_j-a^{-1}z^{\tau})
\prod_{j=1}^{M-1}(az^{-\tau}-a^{-1}z^{\tau}w_j).
\end{align}
\end{proposition}

\section{Symmetric functions}
We introduce symmetric functions in this section and show that they
represent the (dual) wavefunctions
of the $U_q(sl_2)$ six-vertex model under reflecting boundary.
We also compare the homogeneous limit of the symmetric functions
with the coordinate Bethe ansatz wavefunctions
for the open XXZ chain by Alcaraz-Barber-Batchelor-Baxter-Quispel \cite{ABBBQ}.

\begin{definition}
We define the following symmetric function \\
$F_{M,N}(z_1,\dots,z_N|w_1,\dots,w_M|x_1,\dots,x_N)$
which depends on the symmetric variables \\
$z_1,\dots,z_N$,
complex parameters $w_1,\dots,w_M$
and integers $x_1,\dots,x_N$ satisfying
$1 \le x_1 < \cdots < x_N \le M$,
\begin{align}
&F_{M,N}(z_1,\dots,z_N|w_1,\dots,w_M|x_1,\dots,x_N)
\nonumber \\
=&
(a^2-a^{-2})^N \prod_{j=1}^N (-a^2 z_j^2+a^{-2} z_j^{-2})
\sum_{\sigma \in S_N}
\sum_{\tau_1=\pm 1,\dots,\tau_N=\pm 1}
\prod_{j=1}^N \frac{1}{z_j^{2 \tau_j}-z_j^{-2 \tau_j}}
\nonumber \\
&\times \prod_{1 \le j < k \le N}
\frac{(a^2z_{\sigma(j)}^{\tau_{\sigma(j)}}
z_{\sigma(k)}^{\tau_{\sigma(k)}}
-a^{-2}z_{\sigma(j)}^{-\tau_{\sigma(j)}}
z_{\sigma(k)}^{-\tau_{\sigma(k)}}
)
(
a^2z_{\sigma(j)}^{-\tau_{\sigma(j)}}
z_{\sigma(k)}^{\tau_{\sigma(k)}}
-a^{-2}z_{\sigma(j)}^{\tau_{\sigma(j)}}
z_{\sigma(k)}^{-\tau_{\sigma(k)}}
)}
{(
z_{\sigma(j)}^{\tau_{\sigma(j)}}
z_{\sigma(k)}^{\tau_{\sigma(k)}}
-z_{\sigma(j)}^{-\tau_{\sigma(j)}}
z_{\sigma(k)}^{-\tau_{\sigma(k)}}
)
(
z_{\sigma(j)}^{-\tau_{\sigma(j)}}
z_{\sigma(k)}^{\tau_{\sigma(k)}}
-z_{\sigma(j)}^{\tau_{\sigma(j)}}
z_{\sigma(k)}^{-\tau_{\sigma(k)}}
)}
\nonumber \\
&\times
\prod_{k=1}^N (ba^{-1}z_{\sigma(k)}^{-\tau_{\sigma(k)}}-b^{-1}az_{\sigma(k)}^{\tau_{\sigma(k)}})
\prod_{k=1}^N \prod_{j=1}^M (az_{\sigma(k)}^{-\tau_{\sigma(k)}}-a^{-1}z_{\sigma(k)}^{\tau_{\sigma(k)}}w_j)
\nonumber \\
&\times
\prod_{k=1}^N \prod_{j=1}^{x_k-1} (az_{\sigma(k)}^{-\tau_{\sigma(k)}}w_j-a^{-1}z_{\sigma(k)}^{\tau_{\sigma(k)}})
\prod_{k=1}^N \prod_{j=x_k+1}^M (az_{\sigma(k)}^{\tau_{\sigma(k)}}-a^{-1}z_{\sigma(k)}^{-\tau_{\sigma(k)}}w_j).
\label{ordinaryrighthandside}
\end{align}

We also define the following symmetric function
$\overline{F}_{M,N}(z_1,\dots,z_N|w_1,\dots,w_M|\overline{x_1},\dots,\overline{x_N})$
\begin{align}
&\overline{F}_{M,N}(z_1,\dots,z_N|w_1,\dots,w_M|\overline{x_1},\dots,\overline{x_N})
\nonumber \\
=&
(a^2-a^{-2})^N \prod_{j=1}^N (-a^2 z_j^2+a^{-2} z_j^{-2})
\sum_{\sigma \in S_N}
\sum_{\tau_1=\pm 1,\dots,\tau_N=\pm 1}
\prod_{j=1}^N \frac{1}{z_j^{2 \tau_j}-z_j^{-2 \tau_j}}
\nonumber \\
&\times \prod_{1 \le j < k \le N}
\frac{(a^2z_{\sigma(j)}^{\tau_{\sigma(j)}}
z_{\sigma(k)}^{\tau_{\sigma(k)}}
-a^{-2}z_{\sigma(j)}^{-\tau_{\sigma(j)}}
z_{\sigma(k)}^{-\tau_{\sigma(k)}}
)
(
a^2z_{\sigma(j)}^{-\tau_{\sigma(j)}}
z_{\sigma(k)}^{\tau_{\sigma(k)}}
-a^{-2}z_{\sigma(j)}^{\tau_{\sigma(j)}}
z_{\sigma(k)}^{-\tau_{\sigma(k)}}
)}
{(
z_{\sigma(j)}^{\tau_{\sigma(j)}}
z_{\sigma(k)}^{\tau_{\sigma(k)}}
-z_{\sigma(j)}^{-\tau_{\sigma(j)}}
z_{\sigma(k)}^{-\tau_{\sigma(k)}}
)
(
z_{\sigma(j)}^{-\tau_{\sigma(j)}}
z_{\sigma(k)}^{\tau_{\sigma(k)}}
-z_{\sigma(j)}^{\tau_{\sigma(j)}}
z_{\sigma(k)}^{-\tau_{\sigma(k)}}
)} \nonumber \\
&\times
\prod_{k=1}^N (baz_{\sigma(k)}^{\tau_{\sigma(k)}}-b^{-1}a^{-1}z_{\sigma(k)}^{-\tau_{\sigma(k)}})
\prod_{k=1}^N \prod_{j=1}^M (az_{\sigma(k)}^{-\tau_{\sigma(k)}}w_j-a^{-1}z_{\sigma(k)}^{\tau_{\sigma(k)}})
\nonumber \\
&\times
\prod_{k=1}^N \prod_{j=1}^{\overline{x_k}-1} (az_{\sigma(k)}^{-\tau_{\sigma(k)}}-a^{-1}z_{\sigma(k)}^{\tau_{\sigma(k)}}w_j)
\prod_{k=1}^N \prod_{j=\overline{x_k}+1}^M (az_{\sigma(k)}^{\tau_{\sigma(k)}}w_j-a^{-1}z_{\sigma(k)}^{-\tau_{\sigma(k)}}),
\label{dualordinaryrighthandside}
\end{align}
where $\overline{x_1},\dots,\overline{x_N}$ are integers satisfying
$1 \le \overline{x_1} < \cdots < \overline{x_N} \le M$.

\end{definition}

\begin{theorem} \label{maintheoremstatement}
The wavefunctions of the $U_q(sl_2)$ six-vertex model
under reflecting boundary
\\
$W_{M,N}(z_1,\dots,z_N|w_1,\dots,w_M|x_1,\dots,x_N)$
is expressed as the symmetric function \\
$F_{M,N}(z_1,\dots,z_N|w_1,\dots,w_M|x_1,\dots,x_N)$
\begin{align}
W_{M,N}(z_1,\dots,z_N|w_1,\dots,w_M|x_1,\dots,x_N)
=F_{M,N}(z_1,\dots,z_N|w_1,\dots,w_M|x_1,\dots,x_N).
\label{maintheorem}
\end{align}
The dual wavefunctions of the $U_q(sl_2)$ six-vertex model
under reflecting boundary
\\
$\overline{W}_{M,N}(z_1,\dots,z_N|w_1,\dots,w_M|\overline{x_1},\dots,\overline{x_N})$
is expressed as the symmetric function \\
$\overline{F}_{M,N}(z_1,\dots,z_N|w_1,\dots,w_M|\overline{x_1},\dots,\overline{x_N})$
\begin{align}
\overline{W}_{M,N}(z_1,\dots,z_N|w_1,\dots,w_M|\overline{x_1},\dots,\overline{x_N})
=\overline{F}_{M,N}(z_1,\dots,z_N|w_1,\dots,w_M|\overline{x_1},\dots,\overline{x_N}).
\label{dualmaintheorem}
\end{align}
\end{theorem}

\begin{proof}
Let us show \eqref{maintheorem} since
\eqref{dualmaintheorem} can be proved in the same way.
We show that \eqref{ordinaryrighthandside} satisfies all the Properties
in Proposition \ref{ordinarypropertiesfordomainwallboundarypartitionfunction}.
It can be easily checked from the definition
of the symmetric function \eqref{ordinaryrighthandside}
that it satisfies Properties (2) and (5).

To show Property (3), one first notes that the sum of two summands in
\eqref{ordinaryrighthandside} which are labeled by the common
$\sigma \in S_N$ and $\tau_1,\dots,\tau_N$ except $\tau_i \ (i=1,\dots,N)$
are always invariant under $z_i \longleftrightarrow z_i^{-1}$.
Then \eqref{permutationwavefunction} follows by looking at the
overall factor
$\displaystyle (a^2-a^{-2})^N \prod_{j=1}^N (-a^2 z_j^2+a^{-2} z_j^{-2})$
in \eqref{ordinaryrighthandside}.

Property (1) can be seen by noting that the factor
\begin{align}
\prod_{k=1}^N \prod_{j=1}^M (az_{\sigma(k)}^{-\tau_{\sigma(k)}}-a^{-1}z_{\sigma(k)}^{\tau_{\sigma(k)}}w_j)
\prod_{k=1}^N \prod_{j=x_k+1}^M (az_{\sigma(k)}^{\tau_{\sigma(k)}}-a^{-1}z_{\sigma(k)}^{-\tau_{\sigma(k)}}w_j),
\end{align}
in each summand in \eqref{ordinaryrighthandside}
where $w_M$ lives in becomes
\begin{align}
\prod_{k=1}^N \prod_{j=1}^M (az_{\sigma(k)}^{-\tau_{\sigma(k)}}-a^{-1}z_{\sigma(k)}^{\tau_{\sigma(k)}}w_j)
\prod_{k=1}^{N-1} \prod_{j=x_k+1}^M (az_{\sigma(k)}^{\tau_{\sigma(k)}}-a^{-1}z_{\sigma(k)}^{-\tau_{\sigma(k)}}w_j),
\end{align}
when $x_N=M$.

Let us show \eqref{ordinaryrighthandside} satisfies Property (4).
See \cite{Motr} for case of the simpler wavefunctions
without reflecting boundary. The way to prove is basically the same.
First, we consider the case $x_N=M$.
After the substitution $w_M=a^2 z_N^2$, one can see that
only the summands satisfying $\sigma(N)=N$, $\tau_N=+1$ in 
\eqref{ordinaryrighthandside} survive.
Due to this fact, one can rewrite the evaluation of
$F_{M,N}(z_1,\dots,z_N|w_1,\dots,w_M|x_1,\dots,x_N)$ at $w_M=a^2 z_N^2$ as
\begin{align}
&F_{M,N}(z_1,\dots,z_N|w_1,\dots,w_M|x_1,\dots,x_N)|_{w_M=a^2 z_N^2}
=
\frac{(a^2-a^{-2}) (-a^2 z_N^2+a^{-2} z_N^{-2})}{z_N^2-z_N^{-2}} \nonumber \\
&\times (a^2-a^{-2})^{N-1} \prod_{j=1}^{N-1} (-a^2 z_j^2+a^{-2} z_j^{-2})
\sum_{\sigma \in S_{N-1}}
\sum_{\tau_1=\pm 1,\dots,\tau_{N-1}=\pm 1}
\prod_{j=1}^{N-1} \frac{1}{z_j^{2 \tau_j}-z_j^{-2 \tau_j}}
\nonumber \\
&\times
\prod_{j=1}^{N-1}
\frac{(a^2z_{\sigma(j)}^{\tau_{\sigma(j)}}
z_N
-a^{-2}z_{\sigma(j)}^{-\tau_{\sigma(j)}}
z_N^{-1}
)
(
a^2z_{\sigma(j)}^{-\tau_{\sigma(j)}}
z_N
-a^{-2}z_{\sigma(j)}^{\tau_{\sigma(j)}}
z_N^{-1}
)}
{(
z_{\sigma(j)}^{\tau_{\sigma(j)}}
z_{N}
-z_{\sigma(j)}^{-\tau_{\sigma(j)}}
z_{N}^{-1}
)
(
z_{\sigma(j)}^{-\tau_{\sigma(j)}}
z_{N}
-z_{\sigma(j)}^{\tau_{\sigma(j)}}
z_{N}^{-1}
)}
\nonumber \\
&\times \prod_{1 \le j < k \le N-1}
\frac{(a^2z_{\sigma(j)}^{\tau_{\sigma(j)}}
z_{\sigma(k)}^{\tau_{\sigma(k)}}
-a^{-2}z_{\sigma(j)}^{-\tau_{\sigma(j)}}
z_{\sigma(k)}^{-\tau_{\sigma(k)}}
)
(
a^2z_{\sigma(j)}^{-\tau_{\sigma(j)}}
z_{\sigma(k)}^{\tau_{\sigma(k)}}
-a^{-2}z_{\sigma(j)}^{\tau_{\sigma(j)}}
z_{\sigma(k)}^{-\tau_{\sigma(k)}}
)}
{(
z_{\sigma(j)}^{\tau_{\sigma(j)}}
z_{\sigma(k)}^{\tau_{\sigma(k)}}
-z_{\sigma(j)}^{-\tau_{\sigma(j)}}
z_{\sigma(k)}^{-\tau_{\sigma(k)}}
)
(
z_{\sigma(j)}^{-\tau_{\sigma(j)}}
z_{\sigma(k)}^{\tau_{\sigma(k)}}
-z_{\sigma(j)}^{\tau_{\sigma(j)}}
z_{\sigma(k)}^{-\tau_{\sigma(k)}}
)}
\nonumber
\end{align}
\begin{align}
&\times
(ba^{-1}z_N^{-1}-b^{-1}az_N)
\prod_{k=1}^{N-1} (ba^{-1}z_{\sigma(k)}^{-\tau_{\sigma(k)}}-b^{-1}az_{\sigma(k)}^{\tau_{\sigma(k)}}) \nonumber \\
&\times (az_N^{-1}-az_N^3) \prod_{j=1}^{M-1} (az_N^{-1}-a^{-1}z_Nw_j)
\nonumber \\
&\times
\prod_{k=1}^{N-1} (az_{\sigma(k)}^{-\tau_{\sigma(k)}}-
a z_N^2 z_{\sigma(k)}^{\tau_{\sigma(k)}})
\prod_{k=1}^{N-1} \prod_{j=1}^{M-1} (az_{\sigma(k)}^{-\tau_{\sigma(k)}}-a^{-1}z_{\sigma(k)}^{\tau_{\sigma(k)}}w_j)
\nonumber \\
&\times
\prod_{j=1}^{M-1} (az_N^{-1}w_j-a^{-1}z_N)
\prod_{k=1}^{N-1} \prod_{j=1}^{x_k-1} (az_{\sigma(k)}^{-\tau_{\sigma(k)}}w_j-a^{-1}z_{\sigma(k)}^{\tau_{\sigma(k)}})
\nonumber \\
&\times
\prod_{k=1}^{N-1} (az_{\sigma(k)}^{\tau_{\sigma(k)}}-az_N^2
z_{\sigma(k)}^{-\tau_{\sigma(k)}})
\prod_{k=1}^{N-1} \prod_{j=x_k+1}^{M-1} (az_{\sigma(k)}^{\tau_{\sigma(k)}}-a^{-1}z_{\sigma(k)}^{-\tau_{\sigma(k)}}w_j). \label{manyfactors}
\end{align}
One can show by tedious but straightforward computation that
a part of the product of factors in \eqref{manyfactors} given below
\begin{align}
&\frac{(a^2-a^{-2}) (-a^2 z_N^2+a^{-2} z_N^{-2})}{z_N^2-z_N^{-2}}
\prod_{j=1}^{N-1}
\frac{(a^2z_{\sigma(j)}^{\tau_{\sigma(j)}}
z_N
-a^{-2}z_{\sigma(j)}^{-\tau_{\sigma(j)}}
z_N^{-1}
)
(
a^2z_{\sigma(j)}^{-\tau_{\sigma(j)}}
z_N
-a^{-2}z_{\sigma(j)}^{\tau_{\sigma(j)}}
z_N^{-1}
)}
{(
z_{\sigma(j)}^{\tau_{\sigma(j)}}
z_{N}
-z_{\sigma(j)}^{-\tau_{\sigma(j)}}
z_{N}^{-1}
)
(
z_{\sigma(j)}^{-\tau_{\sigma(j)}}
z_{N}
-z_{\sigma(j)}^{\tau_{\sigma(j)}}
z_{N}^{-1}
)}
\nonumber \\
&\times
(ba^{-1}z_N^{-1}-b^{-1}az_N)(az_N^{-1}-az_N^3) \prod_{j=1}^{M-1} (az_N^{-1}-a^{-1}z_Nw_j) \prod_{k=1}^{N-1} (az_{\sigma(k)}^{-\tau_{\sigma(k)}}-
a z_N^2 z_{\sigma(k)}^{\tau_{\sigma(k)}})
\nonumber \\
&\times \prod_{j=1}^{M-1} (az_N^{-1}w_j-a^{-1}z_N)
\prod_{k=1}^{N-1} (az_{\sigma(k)}^{\tau_{\sigma(k)}}-az_N^2
z_{\sigma(k)}^{-\tau_{\sigma(k)}}),
\end{align}
can be simplified as
\begin{align}
&\frac{(a^2-a^{-2}) (-a^2 z_N^2+a^{-2} z_N^{-2})}{z_N^2-z_N^{-2}}
\prod_{j=1}^{N-1}
\frac{(a^2 z_{j}z_{N}-a^{-2}z_{j}^{-1}z_{N}^{-1})(a^2 z_{j}^{-1}z_{N}-a^{-2}z_{j}z_{N}^{-1})}
{(z_{j}z_{N}-z_{j}^{-1}z_{N}^{-1})
(z_{j}^{-1}z_{N}-z_{j}z_{N}^{-1})}
\nonumber \\
&\times
(ba^{-1}z_N^{-1}-b^{-1}az_N)(az_N^{-1}-az_N^3) \prod_{j=1}^{M-1} (az_N^{-1}-a^{-1}z_Nw_j) \prod_{k=1}^{N-1} (az_{k}^{-1}-
a z_N^2 z_{k})
\nonumber \\
&\times \prod_{j=1}^{M-1} (az_N^{-1}w_j-a^{-1}z_N)
\prod_{k=1}^{N-1} (az_{k}-az_N^2
z_{k}^{-1})
\nonumber \\
=&(a^2-a^{-2})(ba^{-1}z_N^{-1}-b^{-1}az_N)
\prod_{j=1}^{N}(a^3 z_N^2 z_j-a^{-1}z_j^{-1})
\prod_{j=1}^{N-1} (a^3 z_N^2 z_j^{-1}-a^{-1}z_j) \nonumber \\
&\times \prod_{j=1}^{M-1} (az_N^{-1}w_j-a^{-1}z_N)(az_N^{-1}-a^{-1}z_N w_j).
\label{simplifiedfactors}
\end{align}
From this simplification, one finds that the right hand side of
\eqref{manyfactors} can be expressed as the product of the right hand side of
\eqref{simplifiedfactors} and
$F_{M-1,N-1}(z_1,\dots,z_{N-1}|w_1,\dots,w_{M-1}|x_1,\dots,x_{N-1})$,
and we get
\begin{align}
&F_{M,N}(z_1,\dots,z_N|w_1,\dots,w_M|x_1,\dots,x_N)
|_{w_M=a^2 z_N^2}
\nonumber \\
=&(a^2-a^{-2})(ba^{-1}z_N^{-1}-b^{-1}az_N)
\prod_{j=1}^{N}(a^3 z_N^2 z_j-a^{-1}z_j^{-1})
\prod_{j=1}^{N-1} (a^3 z_N^2 z_j^{-1}-a^{-1}z_j) \nonumber \\
&\times \prod_{j=1}^{M-1} (az_N^{-1}w_j-a^{-1}z_N)(az_N^{-1}-a^{-1}z_N w_j)
\nonumber \\
&\times F_{M-1,N-1}(z_1,\dots,z_{N-1}|w_1,\dots,w_{M-1}|x_1,\dots,x_{N-1})
,
\end{align}
hence it is shown that $F_{M,N}(z_1,\dots,z_N|w_1,\dots,w_M|x_1,\dots,x_N)$
satisfies Property (4) for the case $x_N=M$.

What remains to check is
Property (4) for the case $x_N \neq M$, which can be proved as follows.
First, we rewrite the symmetric functions
$F_{M,N}(z_1,\dots,z_N|w_1,\dots,w_M|x_1,\dots,x_N)$
for the case $x_N \neq M$ as

\begin{align}
&F_{M,N}(z_1,\dots,z_N|w_1,\dots,w_M|x_1,\dots,x_N)
\nonumber \\
=&
(a^2-a^{-2})^N \prod_{j=1}^N (-a^2 z_j^2+a^{-2} z_j^{-2})
\sum_{\sigma \in S_N}
\sum_{\tau_1=\pm 1,\dots,\tau_N=\pm 1}
\prod_{j=1}^N \frac{1}{z_j^{2 \tau_j}-z_j^{-2 \tau_j}}
\nonumber \\
&\times \prod_{1 \le j < k \le N}
\frac{(a^2z_{\sigma(j)}^{\tau_{\sigma(j)}}
z_{\sigma(k)}^{\tau_{\sigma(k)}}
-a^{-2}z_{\sigma(j)}^{-\tau_{\sigma(j)}}
z_{\sigma(k)}^{-\tau_{\sigma(k)}}
)
(
a^2z_{\sigma(j)}^{-\tau_{\sigma(j)}}
z_{\sigma(k)}^{\tau_{\sigma(k)}}
-a^{-2}z_{\sigma(j)}^{\tau_{\sigma(j)}}
z_{\sigma(k)}^{-\tau_{\sigma(k)}}
)}
{(
z_{\sigma(j)}^{\tau_{\sigma(j)}}
z_{\sigma(k)}^{\tau_{\sigma(k)}}
-z_{\sigma(j)}^{-\tau_{\sigma(j)}}
z_{\sigma(k)}^{-\tau_{\sigma(k)}}
)
(
z_{\sigma(j)}^{-\tau_{\sigma(j)}}
z_{\sigma(k)}^{\tau_{\sigma(k)}}
-z_{\sigma(j)}^{\tau_{\sigma(j)}}
z_{\sigma(k)}^{-\tau_{\sigma(k)}}
)}
\nonumber \\
&\times
\prod_{k=1}^N (ba^{-1}z_{\sigma(k)}^{-\tau_{\sigma(k)}}-b^{-1}az_{\sigma(k)}^{\tau_{\sigma(k)}})
\prod_{k=1}^N (az_{\sigma(k)}^{-\tau_{\sigma(k)}}-a^{-1}z_{\sigma(k)}^{\tau_{\sigma(k)}}w_M) \nonumber \\
&\times
\prod_{k=1}^N \prod_{j=1}^{M-1} (az_{\sigma(k)}^{-\tau_{\sigma(k)}}-a^{-1}z_{\sigma(k)}^{\tau_{\sigma(k)}}w_j)
\prod_{k=1}^N \prod_{j=1}^{x_k-1} (az_{\sigma(k)}^{-\tau_{\sigma(k)}}w_j-a^{-1}z_{\sigma(k)}^{\tau_{\sigma(k)}}) \nonumber \\
&\times \prod_{k=1}^N (az_{\sigma(k)}^{\tau_{\sigma(k)}}-a^{-1}z_{\sigma(k)}^{-\tau_{\sigma(k)}}w_M)
\prod_{k=1}^N \prod_{j=x_k+1}^{M-1} (az_{\sigma(k)}^{\tau_{\sigma(k)}}-a^{-1}z_{\sigma(k)}^{-\tau_{\sigma(k)}}w_j). \label{towardsfactorization}
\end{align}
Noting that the product of factors
\begin{align}
\prod_{k=1}^N (az_{\sigma(k)}^{-\tau_{\sigma(k)}}-a^{-1}z_{\sigma(k)}^{\tau_{\sigma(k)}}w_M)
\prod_{k=1}^N (az_{\sigma(k)}^{\tau_{\sigma(k)}}-a^{-1}z_{\sigma(k)}^{-\tau_{\sigma(k)}}w_M),
\end{align}
in the right hand side of
\eqref{towardsfactorization}
can be rewritten as
\begin{align}
\prod_{j=1}^N (az_j^{-1}-a^{-1}z_j w_M)(az_j-a^{-1}z_j^{-1}w_M),
\label{usethisfactor}
\end{align}
which do not have any dependence on $\sigma$,
one finds that \eqref{towardsfactorization}
can be rewritten as a product of
\eqref{usethisfactor} and
$F_{M-1,N}(z_1,\dots,z_N|w_1,\dots,w_{M-1}|x_1,\dots,x_N)$
\begin{align}
&F_{M,N}(z_1,\dots,z_N|w_1,\dots,w_M|x_1,\dots,x_N)
\nonumber \\
=&\prod_{j=1}^N (az_j^{-1}-a^{-1}z_j w_M)(az_j-a^{-1}z_j^{-1}w_M)
F_{M-1,N}(z_1,\dots,z_N|w_1,\dots,w_{M-1}|x_1,\dots,x_N).
\end{align}
Thus we have proven that $F_{M,N}(z_1,\dots,z_N|w_1,\dots,w_M|x_1,\dots,x_N)$
satisfies Property (4) for the case $x_N \neq M$.
\end{proof}

Now let us compare the symmetric functions
$F_{M,N}(z_1,\dots,z_N|w_1,\dots,w_M|x_1,\dots,x_N)$ \eqref{ordinaryrighthandside}
introduced in this section with
the coordinate Bethe ansatz wavefunctions for the open XXZ chain by
Alcaraz-Barber-Batchelor-Baxter-Quispel \cite{ABBBQ}.
For comparison, we consider the homogeneous limit $w_1=\cdots=w_M=1$
of $F_{M,N}(z_1,\dots,z_N|w_1,\dots,w_M|x_1,\dots,x_N)$

\begin{align}
&F_{M,N}(z_1,\dots,z_N|1,\dots,1|x_1,\dots,x_N)
\nonumber \\
=&
(a^2-a^{-2})^N \prod_{j=1}^N (-a^2 z_j^2+a^{-2} z_j^{-2})
\sum_{\sigma \in S_N}
\sum_{\tau_1=\pm 1,\dots,\tau_N=\pm 1}
\prod_{j=1}^N \frac{1}{z_j^{2 \tau_j}-z_j^{-2 \tau_j}}
\nonumber \\
&\times \prod_{1 \le j < k \le N}
\frac{(a^2z_{\sigma(j)}^{\tau_{\sigma(j)}}
z_{\sigma(k)}^{\tau_{\sigma(k)}}
-a^{-2}z_{\sigma(j)}^{-\tau_{\sigma(j)}}
z_{\sigma(k)}^{-\tau_{\sigma(k)}}
)
(
a^2z_{\sigma(j)}^{-\tau_{\sigma(j)}}
z_{\sigma(k)}^{\tau_{\sigma(k)}}
-a^{-2}z_{\sigma(j)}^{\tau_{\sigma(j)}}
z_{\sigma(k)}^{-\tau_{\sigma(k)}}
)}
{(
z_{\sigma(j)}^{\tau_{\sigma(j)}}
z_{\sigma(k)}^{\tau_{\sigma(k)}}
-z_{\sigma(j)}^{-\tau_{\sigma(j)}}
z_{\sigma(k)}^{-\tau_{\sigma(k)}}
)
(
z_{\sigma(j)}^{-\tau_{\sigma(j)}}
z_{\sigma(k)}^{\tau_{\sigma(k)}}
-z_{\sigma(j)}^{\tau_{\sigma(j)}}
z_{\sigma(k)}^{-\tau_{\sigma(k)}}
)}
\nonumber \\
&\times
\prod_{k=1}^N (ba^{-1}z_{\sigma(k)}^{-\tau_{\sigma(k)}}-b^{-1}az_{\sigma(k)}^{\tau_{\sigma(k)}})
\prod_{k=1}^N (az_{\sigma(k)}^{-\tau_{\sigma(k)}}-a^{-1}z_{\sigma(k)}^{\tau_{\sigma(k)}})^M
\nonumber \\
&\times
\prod_{k=1}^N (az_{\sigma(k)}^{-\tau_{\sigma(k)}}-a^{-1}z_{\sigma(k)}^{\tau_{\sigma(k)}})^{x_k-1}
\prod_{k=1}^N (az_{\sigma(k)}^{\tau_{\sigma(k)}}-a^{-1}z_{\sigma(k)}^{-\tau_{\sigma(k)}})^{M-x_k}. \label{forcomparisonwithcoordinatebetheansatz}
\end{align}
One first rewrites \eqref{forcomparisonwithcoordinatebetheansatz} as
\begin{align}
&F_{M,N}(z_1,\dots,z_N|1,\dots,1|x_1,\dots,x_N)
\nonumber \\
=&
\frac{(a^2-a^{-2})^N}
{\prod_{1 \le j < k \le N}
(z_k^2+z_k^{-2}-z_j^2-z_j^{-2})
}
\prod_{j=1}^N 
\frac{(-a^2 z_j^2+a^{-2} z_j^{-2})(a^2+a^{-2}-z_j^2-z_j^{-2})^M}
{z_j^2-z_j^{-2}} \nonumber \\
&\times \sum_{\sigma \in S_N}
\sum_{\tau_1=\pm 1,\dots,\tau_N=\pm 1} \mathrm{sgn} (\sigma) (-1)^{|\tau|}
\nonumber \\
&\times \prod_{1 \le j < k \le N}
(a^2z_{\sigma(j)}^{\tau_{\sigma(j)}}
z_{\sigma(k)}^{\tau_{\sigma(k)}}
-a^{-2}z_{\sigma(j)}^{-\tau_{\sigma(j)}}
z_{\sigma(k)}^{-\tau_{\sigma(k)}}
)
(
a^2z_{\sigma(j)}^{-\tau_{\sigma(j)}}
z_{\sigma(k)}^{\tau_{\sigma(k)}}
-a^{-2}z_{\sigma(j)}^{\tau_{\sigma(j)}}
z_{\sigma(k)}^{-\tau_{\sigma(k)}}
)
\nonumber \\
&\times
\prod_{k=1}^N \frac{ba^{-1}z_{\sigma(k)}^{-\tau_{\sigma(k)}}-b^{-1}az_{\sigma(k)}^{\tau_{\sigma(k)}}}
{
az_{\sigma(k)}^{-\tau_{\sigma(k)}}-a^{-1}z_{\sigma(k)}^{\tau_{\sigma(k)}}}
\prod_{k=1}^N 
\Bigg(
\frac{
az_{\sigma(k)}^{-\tau_{\sigma(k)}}-a^{-1}z_{\sigma(k)}^{\tau_{\sigma(k)}}}
{az_{\sigma(k)}^{\tau_{\sigma(k)}}-a^{-1}z_{\sigma(k)}^{-\tau_{\sigma(k)}}}
\Bigg)^{x_k},
\label{forcomparisonwithcoordinatebetheansatztwo}
\end{align}
where $|\tau|$ denotes the number of $\tau_j$'s satisfying $\tau_j=-1$.
We introduce the following variables
$\displaystyle \mathrm{e}^{iK_\ell}=\frac{az_\ell-a^{-1}z_\ell^{-1}}{az_\ell^{-1}-a^{-1}z_\ell}, \ell=1,\dots,N,
\Delta=-\frac{a^2+a^{-2}}{2}$, $\displaystyle p^\prime=-\frac{b+b^{-1}}{2(b-b^{-1})}(a^2-a^{-2})$, which are natural
parametrizations for the description of the open XXZ chain
used in \cite{ABBBQ}.
Each represents the momentums, anisotropy parameter and boundary parameter
respectively.
One can show the following relations
\begin{align}
\frac{ba^{-1}z_j^{-1}-b^{-1}az_j}{az_j^{-1}-a^{-1}z_j}
&=\frac{b-b^{-1}}{a^2-a^{-2}}(1+(p^\prime-\Delta) \mathrm{e}^{iK_j}), \label{forcomparisonone} \\
a^2 z_j z_k-a^{-2} z_j^{-1} z_k^{-1}&
=(az_k-a^{-1}z_k^{-1})(az_j^{-1}-a^{-1}z_j)
\frac{1-2\Delta \mathrm{e}^{iK_j}+\mathrm{e}^{iK_j-iK_k}}{a^2-a^{-2}} \label{forcomparisontwo} \\
a^2 z_j^{-1} z_k-a^{-2} z_j z_k^{-1}&=(az_k^{-1}-a^{-1}z_k)(az_j-a^{-1}z_j^{-1})
\frac{1-2\Delta \mathrm{e}^{iK_k}+\mathrm{e}^{iK_j+iK_k}}{a^2-a^{-2}} \mathrm{e}^{-iK_j} \label{forcomparisonthree}.
\end{align}
Using
\eqref{forcomparisonone}, \eqref{forcomparisontwo} and \eqref{forcomparisonthree}, one finds that
\eqref{forcomparisonwithcoordinatebetheansatztwo}
can be rewritten as
\begin{align}
&F_{M,N}(z_1,\dots,z_N|1,\dots,1|x_1,\dots,x_N)
\nonumber \\
=&(b-b^{-1})^N
\prod_{1 \le j < k \le N}
\frac{
(a^2+a^{-2}-z_j^2-z_j^{-2})(a^2+a^{-2}-z_k^{-2}-z_k^{-2})
}
{
(a^2-a^{-2})^2(z_k^2+z_k^{-2}-z_j^2-z_j^{-2})
} \nonumber \\
&\times\prod_{j=1}^N 
\frac{(-a^2 z_j^2+a^{-2} z_j^{-2})(a^2+a^{-2}-z_j^2-z_j^{-2})^M}
{z_j^2-z_j^{-2}} f_{M,N}(K_1,\dots,K_N|x_1,\dots,x_N),
\label{coordinatewavefunctionsappearance}
\end{align}
where
\begin{align}
&f_{M,N}(K_1,\dots,K_N|x_1,\dots,x_N)=
\sum_{\sigma \in S_N}
\sum_{\tau_1=\pm 1,\dots,\tau_N=\pm 1} \mathrm{sgn} (\sigma) (-1)^{|\tau|}
\nonumber \\
&\times \prod_{1 \le j < k \le N}
(1-2 \Delta \mathrm{e}^{i \tau_{\sigma(j)} K_{\sigma(j)}}
+\mathrm{e}^{i \tau_{\sigma(j)} K_{\sigma(j)}-i \tau_{\sigma(k)} K_{\sigma(k)}}
) \nonumber \\
&\times(
1-2 \Delta \mathrm{e}^{i \tau_{\sigma(k)} K_{\sigma(k)}}
+\mathrm{e}^{i \tau_{\sigma(j)} K_{\sigma(j)}+i \tau_{\sigma(k)} K_{\sigma(k)}}
)
\mathrm{e}^{-i \tau_{\sigma(j)} K_{\sigma(j)}}
\nonumber \\
&\times
\prod_{k=1}^N \mathrm{e}^{-i \tau_{\sigma(k)} K_{\sigma(k)} x_k} (1+(p^\prime-\Delta)
\mathrm{e}^{i \tau_{\sigma(k)} K_{\sigma(k)}}).
\end{align}
$f_{M,N}(K_1,\dots,K_N|x_1,\dots,x_N)$ can be expressed
in a compact way
\begin{align}
&f_{M,N}(K_1,\dots,K_N|x_1,\dots,x_N) \nonumber \\
=&
\sum_{P} \epsilon_P
\prod_{1 \le j < k \le N}
(1-2 \Delta \mathrm{e}^{i K_{j}}
+\mathrm{e}^{i K_{j}-i K_{k}}
)
(
1-2 \Delta \mathrm{e}^{i K_{k}}
+\mathrm{e}^{i K_{j}+i K_{k}}
)
\mathrm{e}^{-i K_{j}}
\nonumber \\
&\times
\prod_{k=1}^N \mathrm{e}^{-i K_{k} x_k} (1+(p^\prime-\Delta)
\mathrm{e}^{i K_{k}}), \label{before}
\end{align}
where the sum means that we take sum over all
permutations and negations of $K_1, \dots, K_N$,
and $\epsilon_P$ changes sign at each such ``mutation".
Relabelling the momentums and positions of down spins as
$K^\prime_j=K_{N+1-j}$,
$x^\prime_j=M+1-x_{N+1-j}$ ($j=1,\dots,N$), \eqref{before} can be rewritten as
\begin{align}
&f_{M,N}(K_1,\dots,K_N|x_1,\dots,x_N) \nonumber \\
=&
\sum_{P} \epsilon_P
\prod_{k=1}^N (1+(p^\prime-\Delta)
\mathrm{e}^{i K^\prime_{k}}) \mathrm{e}^{-i(M+1) K^\prime_{k}}
\nonumber \\
&\times \prod_{1 \le j < k \le N}
(1-2 \Delta \mathrm{e}^{i K^\prime_{k}}
+\mathrm{e}^{i K^\prime_{k}-i K^\prime_{j}}
)
(
1-2 \Delta \mathrm{e}^{i K^\prime_{j}}
+\mathrm{e}^{i K^\prime_{j}+i K^\prime_{k}}
)
\mathrm{e}^{-i K^\prime_{k}}
\prod_{k=1}^N \mathrm{e}^{i K^\prime_{k} x^\prime_k} \nonumber \\
=&\sum_P \epsilon_P A(K^\prime_1,\dots,K^\prime_N) \mathrm{e}^{\sum_{k=1}^N i K^\prime_{k} x^\prime_k}, \label{equationABBBQ}
\end{align}
where
\begin{align}
A(K^\prime_1,\dots,K^\prime_N)=&\prod_{j=1}^N \beta(-K^\prime_j)
\prod_{1 \le j < k \le N} B(-K^\prime_j,K^\prime_k)\mathrm{e}^{-iK^\prime_k},
\\
\beta(K^\prime)=&(1+(p^\prime-\Delta)
\mathrm{e}^{-i K^\prime}) \mathrm{e}^{i(M+1) K^\prime}, \\
B(-K^\prime_j,K^\prime_k)=&s(-K^\prime_j,K^\prime_k)s(K^\prime_k,K^\prime_j),
\\
s(K^\prime_1,K^\prime_2)=&1-2 \Delta \mathrm{e}^{iK^\prime_2}+
\mathrm{e}^{i(K^\prime_1+K^\prime_2)}.
\end{align}
$f_{M,N}(K_1,\dots,K_N|x_1,\dots,x_N)$
\eqref{equationABBBQ}, whose relation with
the symmetric functions \\
$F_{M,N}(z_1,\dots,z_N|1,\dots,1|x_1,\dots,x_N)$ is given by
\eqref{coordinatewavefunctionsappearance},
is the form of the coordinate Bethe ansatz wavefunctions
for the open XXZ chain by Alcaraz-Barber-Batchelor-Baxter-Quispel
(equations (2.33), (2.34) in \cite{ABBBQ}).

\section{Algebraic identities}
In this section,
we derive algebraic identities for the symmetric functions
introduced in the last section as an application
of the correspondence between the wavefunctions and symmetric functions
proven in the last section.
We combine the result in the last section with that
on the domain wall boundary partition function by Tsuchiya \cite{Tsuchiya}
and Kuperberg \cite{Ku2}.
The domain wall boundary partition function under reflecting boundary
$Z_M(z_1,\dots,z_M|w_1,\dots,w_M)$ is a special case
$M=N$, $x_j=j$, $j=1,\dots,M$ of the wavefunction under reflecting boundary
\begin{align}
Z_M(z_1,\dots,z_M|w_1,\dots,w_M)
=W_{M,M}(z_1,\dots,z_M|w_1,\dots,w_M|1,\dots,M).
\end{align}

First, let us recall the determinant formula of the domain wall boundary
partition functions under reflecting boundary.

\begin{theorem} (Tsuchiya \cite{Tsuchiya}, Kuperberg \cite{Ku2}) \label{inhomogeneousexpression}
The domain wall boundary partition function under reflecting boundary
$Z_M(z_1,\dots,z_M|w_1,\dots,w_M)$ can be expressed as the
following determinant

\begin{align}
&Z_M(z_1,\dots,z_M|w_1,\dots,w_M)
=(a^2-a^{-2})^M \prod_{i=1}^M w_i^M \prod_{i=1}^M (bw_i^{-1}-b^{-1})
(a^2 z_i^2-a^{-2} z_i^{-2}) \nonumber \\
&\times\frac{
\prod_{i,j=1}^M
(a^2+a^{-2}-z_i^{-2}w_j-z_i^2 w_j^{-1})(a^2+a^{-2}-z_i^{-2}w_j^{-1}-z_i^2 w_j)
}{\prod_{1 \le i<j \le M}
(-z_i^{-1}z_j+z_i z_j^{-1})(z_i^{-1}z_j^{-1}-z_i z_j)
(-w_j^{-1}+w_i^{-1})(w_i w_j-1)
} \nonumber \\
&\times \mathrm{det}_M \Bigg( \frac{1}
{
(a^2+a^{-2}-z_i^{-2}w_j-z_i^2 w_j^{-1})(a^2+a^{-2}-z_i^{-2}w_j^{-1}-z_i^2 w_j)
} \Bigg). \label{determinantinhomogeneheous}
\end{align}

\end{theorem}

\begin{theorem} \label{inhomogeneouspairing}
We have the following algebraic identities for the
symmetric functions \\
$F(z_{M-N+1},\dots,z_M|w_1,\dots,w_M|x_1,\dots,x_N)$
and $\overline{F}(z_1,\dots,z_{M-N}|w_1,\dots,w_M|\overline{x_1},\dots,\overline{x_{M-N}})$
\begin{align}
&\sum_{x \sqcup \overline{x}=\{1,2,\dots,M \}} \overline{F}(z_1,\dots,z_{M-N}|w_1,\dots,w_M|\overline{x_1},\dots,\overline{x_{M-N}}) \nonumber \\
&\times F(z_{M-N+1},\dots,z_M|w_1,\dots,w_M|x_1,\dots,x_N)
\nonumber \\
=&(a^2-a^{-2})^M \prod_{i=1}^M w_i^M \prod_{i=1}^M (bw_i^{-1}-b^{-1})
(a^2 z_i^2-a^{-2} z_i^{-2}) \nonumber \\
&\times\frac{
\prod_{i,j=1}^M
(a^2+a^{-2}-z_i^{-2}w_j-z_i^2 w_j^{-1})(a^2+a^{-2}-z_i^{-2}w_j^{-1}-z_i^2 w_j)
}{\prod_{1 \le i<j \le M}
(-z_i^{-1}z_j+z_i z_j^{-1})(z_i^{-1}z_j^{-1}-z_i z_j)
(-w_j^{-1}+w_i^{-1})(w_i w_j-1)
} \nonumber \\
&\times \mathrm{det}_M \Bigg( \frac{1}
{
(a^2+a^{-2}-z_i^{-2}w_j-z_i^2 w_j^{-1})(a^2+a^{-2}-z_i^{-2}w_j^{-1}-z_i^2 w_j)
} \Bigg).
\label{pairing}
\end{align}
Here, the sum $\displaystyle \sum_{x \sqcup \overline{x}=\{1,2,\dots,M \}}$ means that
we take the sum over $x=\{ x_1,\dots, x_N \}$ $(1 \le x_1 < x_2 < \cdots < x_N \le M)$ and
$\overline{x}=\{ \overline{x_1} \dots \overline{x_{M-N}} \}$
$(1 \le \overline{x_1} < \overline{x_2} < \cdots < \overline{x_{M-N}} \le M)$
which forms a disjoint union of $\{1,2,\dots,M \}$,
$x \sqcup \overline{x}=\{1,2\dots,M \}$.
\end{theorem}

\begin{proof}
This identity is a consequence of two ways of evaluation
of the domain wall boundary partition function under reflecting boundary
$Z_M(z_1,\dots,z_M|w_1,\dots,w_M)$.
First,
\eqref{determinantinhomogeneheous} in Theorem \ref{inhomogeneousexpression}
gives a direct determinant representation of
$Z_M(z_1,\dots,z_M|w_1,\dots,w_M)$.

Another way to evaluate the domain wall boundary partition function
is to insert the completeness relation
\begin{align}
\sum_{\{ x \}}|x_1 \cdots x_N \rangle \langle x_1 \cdots x_N |=\mathrm{Id},
\end{align}
between the $B$-operators and
using the explicit representations of the
the wavefunctions and its dual
\eqref{maintheorem} and \eqref{dualmaintheorem}
to get
\begin{align}
&Z_M(z_1,\dots,z_M|w_1,\dots,w_M) \nonumber \\
=&\langle 1^M|
\mathcal{B}(z_1|w_1,\dots,w_M) \cdots \mathcal{B}(z_M|w_1,\dots,w_M)|0^M \rangle \nonumber \\
=&\sum_{\{x \}} \langle 1^M|\mathcal{B}(z_1|w_1,\dots,w_M) \cdots
\mathcal{B}(z_{M-N}|w_1,\dots,w_M) |x_1 \cdots x_N \rangle \nonumber \\
&\times \langle x_1 \cdots x_N|
\mathcal{B}(z_{M-N+1}|w_1,\dots,w_M) \cdots
\mathcal{B}(z_M|w_1,\dots,w_M)|0^M \rangle \nonumber \\
=&\sum_{x \sqcup \overline{x}=\{1,\dots,M \}} \langle 1^M|\mathcal{B}(z_1|w_1,\dots,w_M) \cdots
\mathcal{B}(z_{M-N}|w_1,\dots,w_M) |\overline{x_1} \cdots \overline{x_{M-N}} \rangle
\nonumber \\
&\times \langle x_1 \cdots x_N|
\mathcal{B}(z_{M-N+1}|w_1,\dots,w_M) \cdots
\mathcal{B}(z_M|w_1,\dots,w_M)|0^M \rangle \nonumber \\
=&\sum_{x \sqcup \overline{x}=\{1,2,\dots,M \}} \overline{W}(z_1,\dots,z_{M-N}|w_1,\dots,w_M|\overline{x_1},\dots,\overline{x_{M-N}}) \nonumber \\
&\times W(z_{M-N+1},\dots,z_M|w_1,\dots,w_M|x_1,\dots,x_N),
\nonumber \\
=&\sum_{x \sqcup \overline{x}=\{1,2,\dots,M \}} \overline{F}(z_1,\dots,z_{M-N}|w_1,\dots,w_M|\overline{x_1},\dots,\overline{x_{M-N}}) \nonumber \\
&\times F(z_{M-N+1},\dots,z_M|w_1,\dots,w_M|x_1,\dots,x_N).
\label{comparisontwo}
\end{align}
Comparing the two ways of evaluations
\eqref{comparisontwo} and
\eqref{determinantinhomogeneheous}
gives the identity \eqref{pairing}.
\end{proof}

Let us also examine the homogeneous limit $w_1=\cdots=w_M=1$.
First, one can show the following.

\begin{theorem} \label{homogeneousexpression}
In the homogeneous limit $w_1=\cdots=w_M=1$,
the domain wall boundary partition function under reflecting boundary
$Z_M(z_1,\dots,z_M|1,\dots,1)$ can be expressed in the following form:
\begin{align}
&Z_M(z_1,\dots,z_M|1,\dots,1)
=a^{2M} (b-b^{-1})^M \prod_{i=1}^M
\frac{a^2 z_i^2-a^{-2} z_i^{-2}}{1-z_i^{-4}} \nonumber \\
&\times \frac{
\prod_{i=1}^M
(a^2+a^{-2}-z_i^{-2}-z_i^2)^{2M}
}{
\prod_{1 \le i<j \le M}
(-z_i^{-1}z_j+z_i z_j^{-1})(z_i^{-1}z_j^{-1}-z_i z_j)
}
\mathrm{det}_M \Bigg(
\frac{(a^2 z_i^2)^{j-1}}{(a^2-z_i^{2})^{2j}}
-\frac{(a^2 z_i^2)^{j-1}}{(1-a^2 z_i^2)^{2j}}
\Bigg).
\end{align}
\end{theorem}

\begin{proof}

We divide the inhomogeneous determinant of
$Z_M(z_1,\dots,z_M|w_1,\dots,w_M)$
\eqref{determinantinhomogeneheous}
into two parts as
\begin{align}
&Z_M(z_1,\dots,z_M|w_1,\dots,w_M)=P_1 P_2, \nonumber \\
&P_1=(a^2-a^{-2})^M \prod_{i=1}^M w_i^M \prod_{i=1}^M (bw_i^{-1}-b^{-1})
(a^2 z_i^2-a^{-2} z_i^{-2}) \nonumber \\
&\times \frac{
\prod_{i,j=1}^M
(a^2+a^{-2}-z_i^{-2}w_j-z_i^2 w_j^{-1})(a^2+a^{-2}-z_i^{-2}w_j^{-1}-z_i^2 w_j)
}{
\prod_{1 \le i<j \le M}
(-z_i^{-1}z_j+z_i z_j^{-1})(z_i^{-1}z_j^{-1}-z_i z_j)
}, \nonumber \\
&P_2=\frac{
1}{\prod_{1 \le i<j \le M}
(-w_j^{-1}+w_i^{-1})(w_i w_j-1)
} \nonumber \\
&\times \mathrm{det}_M \Bigg( \frac{1}
{
(a^2+a^{-2}-z_i^{-2}w_j-z_i^2 w_j^{-1})(a^2+a^{-2}-z_i^{-2}w_j^{-1}-z_i^2 w_j)
} \Bigg).
\end{align}
It is easy to take the homogeneous limit of $P_1$
\begin{align}
\lim_{w_1 \to 1,\dots,w_M \to 1}P_1
&=
(a^2-a^{-2})^M (b-b^{-1})^M \prod_{i=1}^M
(a^2 z_i^2-a^{-2} z_i^{-2}) \nonumber \\
&\times \frac{
\prod_{i=1}^M
(a^2+a^{-2}-z_i^{-2}-z_i^2)^{2M}
}{
\prod_{1 \le i<j \le M}
(-z_i^{-1}z_j+z_i z_j^{-1})(z_i^{-1}z_j^{-1}-z_i z_j)
}. \label{tochulimit1}
\end{align}
Next, we examine the limit of $P_2$.
We rewrite $P_2$ in terms of $W_j=(w_j+w_j^{-1})/2$ as
\begin{align}
P_2&=\frac{
1}{\prod_{1 \le i<j \le M}
2(W_j-W_i)
} \mathrm{det}_M \Bigg( \frac{a^4 z_i^4}
{
(a^4+z_i^{4}-2a^2 z_i^2 W_j)(1+a^4 z_i^4-2a^2 z_i^2 W_j)
} \Bigg) \nonumber \\
&=\frac{
1}{\prod_{1 \le i<j \le M}
2(W_j-W_i)
} \nonumber \\
&\times \mathrm{det}_M \Bigg( \frac{a^4 z_i^4}{(a^4-1)(z_i^4-1)}
\Bigg(
\frac{1}{a^4+z_i^{4}-2a^2 z_i^2 W_j}
-\frac{1}{1+a^4 z_i^4-2a^2 z_i^2 W_j}
\Bigg)
\Bigg)
\nonumber \\
&=2^{-M(M-1)/2}
\frac{a^{4M} \prod_{i=1}^M z_i^4}{(a^4-1)^{M} \prod_{i=1}^M (z_i^4-1)}
\frac{
1}{\prod_{1 \le i<j \le M}
(W_j-W_i)
} \nonumber \\
&\times \mathrm{det}_M \Bigg(
\frac{1}{a^4+z_i^{4}-2a^2 z_i^2 W_j}
-\frac{1}{1+a^4 z_i^4-2a^2 z_i^2 W_j}
\Bigg)
.
\end{align}
In this form, we can take the homogeneous limit in the same way with
Izergin-Coker-Korepin \cite{ICK} to get
\begin{align}
\lim_{W_1,\dots,W_M \to 1} P_2&=
2^{-M(M-1)/2}
\frac{a^{4M} \prod_{i=1}^M z_i^4}{(a^4-1)^{M} \prod_{i=1}^M (z_i^4-1)}
\nonumber \\
&\times \mathrm{det}_M \Bigg(
\frac{(2a^2 z_i^2)^{j-1}}{(a^4+z_i^{4}-2a^2 z_i^2)^j}
-\frac{(2a^2 z_i^2)^{j-1}}{(1+a^4 z_i^4-2a^2 z_i^2)^j}
\Bigg) \nonumber \\
&=
\frac{a^{4M} \prod_{i=1}^M z_i^4}{(a^4-1)^{M} \prod_{i=1}^M (z_i^4-1)}
\mathrm{det}_M \Bigg(
\frac{(a^2 z_i^2)^{j-1}}{(a^2-z_i^{2})^{2j}}
-\frac{(a^2 z_i^2)^{j-1}}{(1-a^2 z_i^2)^{2j}}
\Bigg).
\label{tochulimit2}
\end{align}
Combining \eqref{tochulimit1} and \eqref{tochulimit2}, we have
\begin{align}
&Z_M(z_1,\dots,z_M|1,\dots,1)
=a^{2M} (b-b^{-1})^M \prod_{i=1}^M
\frac{a^2 z_i^2-a^{-2} z_i^{-2}}{1-z_i^{-4}} \nonumber \\
&\times \frac{
\prod_{i=1}^M
(a^2+a^{-2}-z_i^{-2}-z_i^2)^{2M}
}{
\prod_{1 \le i<j \le M}
(-z_i^{-1}z_j+z_i z_j^{-1})(z_i^{-1}z_j^{-1}-z_i z_j)
}
\mathrm{det}_M \Bigg(
\frac{(a^2 z_i^2)^{j-1}}{(a^2-z_i^{2})^{2j}}
-\frac{(a^2 z_i^2)^{j-1}}{(1-a^2 z_i^2)^{2j}}
\Bigg).
\end{align}
\end{proof}

From Theorem \ref{homogeneousexpression},
we get the following algebraic identities
as a limit of Theorem \ref{inhomogeneouspairing}.

\begin{theorem}
We have the following algebraic identities for the
symmetric functions \\
$F(z_{M-N+1},\dots,z_M|1,\dots,1|x_1,\dots,x_N)$
and $\overline{F}(z_1,\dots,z_{M-N}|1,\dots,1|\overline{x_1},\dots,\overline{x_{M-N}})$
\begin{align}
&\sum_{x \sqcup \overline{x}=\{1,2,\dots,M \}} \overline{F}(z_1,\dots,z_{M-N}|1,\dots,1|\overline{x_1},\dots,\overline{x_{M-N}})
F(z_{M-N+1},\dots,z_M|1,\dots,1|x_1,\dots,x_N)
\nonumber \\
&=
a^{2M} (b-b^{-1})^M \prod_{i=1}^M
\frac{a^2 z_i^2-a^{-2} z_i^{-2}}{1-z_i^{-4}} \nonumber \\
&\times \frac{
\prod_{i=1}^M
(a^2+a^{-2}-z_i^{-2}-z_i^2)^{2M}
}{
\prod_{1 \le i<j \le M}
(-z_i^{-1}z_j+z_i z_j^{-1})(z_i^{-1}z_j^{-1}-z_i z_j)
}
\mathrm{det}_M \Bigg(
\frac{(a^2 z_i^2)^{j-1}}{(a^2-z_i^{2})^{2j}}
-\frac{(a^2 z_i^2)^{j-1}}{(1-a^2 z_i^2)^{2j}}
\Bigg).
\end{align}
Here, the sum $\displaystyle \sum_{x \sqcup \overline{x}=\{1,2,\dots,M \}}$ means that
we take the sum over $x=\{ x_1,\dots, x_N \}$ $(1 \le x_1 < x_2 < \cdots < x_N \le M)$ and
$\overline{x}=\{ \overline{x_1} \dots \overline{x_{M-N}} \}$
$(1 \le \overline{x_1} < \overline{x_2} < \cdots < \overline{x_{M-N}} \le M)$
which forms a disjoint union of $\{1,2,\dots,M \}$,
$x \sqcup \overline{x}=\{1,2\dots,M \}$.
\end{theorem}

\section{Conclusion}
In this paper, we extended the recently developed
Izergin-Korepin analysis on the wavefunctions \cite{Moellipticfelderhof,MoAMP,Motr}
to the $U_q(sl_2)$ six-vertex model with reflecting end.
We determined the exact forms of the symmetric functions
representing the wavefunctions and its dual
based on the Izergin-Korepin analysis.
We also compared the homogeneous limit of the symmetric functions
with the coordinate Bethe ansatz wavefunctions
for the open XXZ chain by Alcaraz-Barber-Batchelor-Baxter-Quispel \cite{ABBBQ}.
As an application of the correspondence between the wavefunctions
and the symmetric functions, we have derived algebraic identities
for the symmetric functions by using the determinant formula for
the domain wall boundary partition functions with reflecting end
by Tsuchiya \cite{Tsuchiya} and Kuperberg \cite{Ku2}.
This idea was first used in Bump-McNamara-Nakasuji \cite{BMN}
to derive dual Cauchy identities for the factorial Schur functions
from the wavefunctions and the domain wall boundary partition functions
of the free-fermionic six-vertex model, and applied to
the $U_q(sl_2)$ six-vertex model in \cite{Motegi} to derive
algebraic identities for the quantum group deformation
of the Grothendieck polynomials.
A different way of using the domain wall boundary
partition functions can be seen in the paper by
Wheeler-Zinn-Justin \cite{WZnew}.

A similar Izergin-Korepin analysis can be done to study the 
free-fermionic six-vertex model under reflecting boundary \cite{Mopreparation}.
It is interesting to extend the analysis to other types
of boundary conditions.
For example, it seems that it is a more challenging task
to treat the wavefunctions under half-turn boundary condition.
At the level of the domain wall boundary partition functions,
the level of difficulty of treating the half-turn boundary and
the reflecting boundary seems to be the same,
but the half-turn boundary condition becomes more difficult
to treat when it comes to problem of the wavefunctions,
since it seems that there is no way to freeze two rows at once
for the case of the wavefunctions under half-turn boundary condition,
contrary to the reflecting boundary condition which is treated in this paper.

Another interesting topic is to study the thermodynamic limit.
As for the domain wall boundary partition functions,
many determinant or Pfaffian formulae are known.
However, it seems hard to expect that the wavefunctions in general
have such simple forms. It is an interesting topic to construct
a new way to study thermodynamic limit without resorting
to the deteminant or Pfaffian formulae.

\section*{Acknowledgments}
The author thanks the referee for helpful comments and suggestions
to improve the paper.
This work was partially supported by grant-in-Aid
for Scientific Research (C) No. 16K05468.


\begin{thebibliography}{00}
%
\bibitem{Baxter}
R.J. Baxter,
{\it Exactly Solved Models in Statistical Mechanics}
(Academic Press, London, 1982).
%
\bibitem{KBI}
V.E. Korepin, N.M. Bogoliubov and A.G. Izergin
{\it Quantum Inverse Scattering Method and Correlation functions}
(Cambridge University Press, Cambridge, 1993).
%
\bibitem{Reshetikhin}
N. Reshetikhin,
{\it Lectures on integrable models in statistical mechanics.}
In: Exact Methods in LowDimensional Statistical Physics and Quantum Computing, Proceedings of Les Houches School in Theoretical Physics. 
(Oxford University Press, 2010).
%
\bibitem{Ko}
V.E. Korepin,
Commun. Math. Phys. {\bf 86}, 391 (1982).
%
\bibitem{Iz}
A. Izergin,
Sov. Phys. Dokl. {\bf 32}, 878 (1987).
%
\bibitem{Br}
D. Bressoud,
{\it Proofs and confirmations:
The story of the alternating sign matrix conjecture}
(MAA Spectrum,
Mathematical Association of America,
Washington, DC, 1999).
%
\bibitem{Ku1}
G. Kuperberg,
Int. Math. Res. Not. {\bf 3}, 123 (1996).
%
\bibitem{Ku2}
G. Kuperberg,
Ann. Math. {\bf 156}, 835 (2002).
%
\bibitem{Okada}
S. Okada,
J. Alg. Comb. {\bf 23}, 43 (2001).
%
\bibitem{Tsuchiya}
O. Tsuchiya,
J. Math. Phys. {\bf 39}, 5946 (1998).
%
\bibitem{ZZfelderhof}
S-Y. Zhao and Y-Z. Zhang,
J. Math. Phys. {\bf 48}, 023504 (2007).
%
\bibitem{FCWZfelderhof}
O. Foda, A.D. Caradoc, M. Wheeler and M.L. Zuparic,
J. Stat. Mech. {\bf 0703,} P03010 (2007).
%
\bibitem{FWZfelderhof}
O. Foda, M. Wheeler and M. Zuparic,
J. Stat. Mech. P02001 (2008).
%
\bibitem{PRS}
S. Pakuliak, V. Rubtsov and A. Silantyev,
J. Phys. A:Math. Theor. {\bf 41}, 295204 (2008).
%
\bibitem{Ros}
H. Rosengren,
Adv. Appl. Math. {\bf 43}, 137 (2009).
%
\bibitem{FK}
F. Filali and N. Kitanine,
J. Stat. Mech. L06001 (2010).
%
\bibitem{Ga}
A. Garbali,
J. Stat. Mech. 033112 (2016).
%
\bibitem{Wheeler}
M. Wheeler,
Nucl. Phys. B {\bf 852}, 469 (2011).
%
\bibitem{Moellipticfelderhof}
K. Motegi,
Prog. Theor. Exp. Phys. {\bf 2017}, 123A01 (2017).
%
\bibitem{MoAMP}
K. Motegi, Adv. Math. Phys. 7563781 (2017).
%
\bibitem{Motr}
K. Motegi, J. Math. Phys. {\bf 59}, 053505 (2018).
%
\bibitem{Bogo}
N.M. Bogoliubov,
J. Phys. A: Math. Gen. {\bf 38}, 9415 (2005).
%
\bibitem{BW}
D. Betea and M. Wheeler,
J. Comb. Th. Ser. A {\bf 137}, 126 (2016).
%
\bibitem{BWZ}
D. Betea, M. Wheeler and P. Zinn-Justin,
J. Alg. Comb. {\bf 42}, 555 (2015).
%
\bibitem{WZnew}
M. Wheeler and P. Zinn-Justin,
Adv. Math. {\bf 299}, 543 (2016).
%
\bibitem{vDE}
J.F. van Diejen and E. Emsiz,
Commun. Math. Phys. {\bf 350}, 1017 (2017).
%
\bibitem{MS}
K. Motegi and K. Sakai,
J. Phys. A: Math. Theor. {\bf 46}, 355201 (2013).
%
\bibitem{Motegi}
K. Motegi, J. Math. Phys. {\bf 58}, 091703 (2017).
%
\bibitem{MS2}
K. Motegi and K. Sakai,
J. Phys. A: Math. Theor. {\bf 47}, 445202 (2014).
%
\bibitem{Korff}
C. Korff,
Lett. Math. Phys. {\bf 104}, 771 (2014).
%
\bibitem{GK2}
V. Gorbounov and C. Korff,
Adv. Math. {\bf 313}, 282 (2017).
%
\bibitem{Borodin}
A. Borodin,
Adv. in Math. {\bf 306}, 973 (2017).
%
\bibitem{BP1}
A. Borodin and L. Petrov,
Sel. Math. New Ser. {\bf 24}, 1 (2018).
%
\bibitem{TakeyamaHecke}
Y. Takeyama,
J. Phys. A {\bf 47}, 465203 (2014).
%
\bibitem{Takeyama}
Y. Takeyama,
arXiv:1606.00578.
%
\bibitem{WZ}
%
M. Wheeler and P. Zinn-Justin,
arXiv:1607.02396, J. Reine Angew. doi:10.1515/crelle-2017-0033.
%
\bibitem{BBF}
B. Brubaker, D. Bump and S. Friedberg,
Commun. Math. Phys. {\bf 308}, 281 (2011).
%
\bibitem{Iv}
D. Ivanov,
{\it Symplectic Ice.} in:
Multiple Dirichlet series, $L$-functions
and automorphic forms, vol 300 of Progr. Math. 
(Birkh\"auser/Springer,
New York, 2012) pp. 205-222.
%
\bibitem{BBCG}
B. Brubaker, D. Bump, G. Chinta and P.E. Gunnells P E,
{\it Metaplectic Whittaker Functions and Crystals of Type B.} in:
Multiple Dirichlet series, $L$-functions
and automorphic forms, vol 300 of Progr. Math. 
(Birkh\"auser/Springer,
New York, 2012) pp. 93-118.
%
\bibitem{Tabony}
S.J. Tabony
{\it Deformations of characters, metaplectic Whittaker functions
and the Yang-Baxter equation,} PhD. Thesis,
Massachusetts Institute of Technology, USA 2011.
%
\bibitem{BMN}
D. Bump, P. McNamara and M. Nakasuji,
Comm. Math. Univ. St. Pauli {\bf 63}, 23 (2014).
%
\bibitem{BSfelderhof}
B. Brubaker and A. Schultz,
J. Alg. Comb. {\bf 42}, 917 (2015).
%
\bibitem{LMP}
K. Motegi,
Lett. Math. Phys. {\bf 107}, 1235 (2017).
%
\bibitem{LS}
A. Lascoux and M. Sch\"utzenberger,
C. R. Acad. Sci. Parix S\'er. I Math
{\bf 295}, 629 (1982).
%
\bibitem{FoKi}
S. Fomin and A.N. Kirillov
Proc. 6th Internat. Conf. on Formal Power Series and
Algebraic Combinatorics, DIMACS 183 (1994).
%
\bibitem{Buch}
A.S. Buch,
Acta. Math. {\bf 189}, 37 (2002).
%
\bibitem{IN2}
T. Ikeda and H. Naruse
Adv. in Math. {\bf 243}, 22 (2013).
%
\bibitem{BBBfelderhof}
B. Brubaker, V. Buciumas and D. Bump,
arXiv:1604.02206,
to appear in \textit{Communications in Number Theory and Physics}.
%
\bibitem{DMfelderhof}
T. Deguchi, and P. Martin,
Int. J. Mod. Phys. A {\it 7}, Suppl. 1A, 165 (1992).
%
\bibitem{RK}
G.A.P. Ribeiro and V.E. Korepin,
J. Phys. A: Math. and Theor. {\bf 48}, 045205 (2015).
%
\bibitem{KZ}
V. Korepin and P. Zinn-Justin,
J. Phys. A {\bf 33}, 7053 (2000).
%
\bibitem{CMRV}
N. Crampe, K. Mallick, E. Ragoucy and M. Vanicat,
J. Phys. A: Math. and Theor. {\bf 48}, 484002 (2015).
%
\bibitem{Fra}
R. Frassek,
J. Phys. A {\bf 50}, 26 (2017).
%
\bibitem{RSV}
N. Reshetikhin, J. Stokman and B. Vlaar,
Comm. Math. Phys. {\bf 336}, 953 (2015).
%
\bibitem{ABBBQ}
F. Alcaraz, M. Barber, M. Batchelor, R. Baxter and G. Quispel,
J. Phys. A: Math. and Gen. {\bf 20} 6397 (1987).
%
\bibitem{Dr}
V. Drinfeld,
Sov. Math. Dokl. {\bf 32}, 254 (1985).
%
\bibitem{J}
M. Jimbo,
Lett. Math. Phys. {\bf 10}, 63 (1985).
%
\bibitem{Sklyanin}
E. Sklyanin,
J. Phys. A: Math. and Gen. {\bf 21}, 2375 (1988).
%
\bibitem{dualsymplecticfelderhof}
K. Motegi,
Rep. Math. Phys. {\bf 80}, 391 (2017).
%
\bibitem{ICK}
A.G. Izergin, D.A. Coker and V.E. Korepin,
J. Phys. A {\bf 25}, 4315 (1992). 
%
\bibitem{Mopreparation}
K. Motegi, in preparation.


\end{thebibliography}
\end{document}